\documentclass[preprint,11pt,sort&compress]{elsarticle}
\usepackage[a4paper,left=2.2cm,right=2.2cm,top=2.5cm,bottom=2.8cm]{geometry}
\usepackage{amsmath}
\usepackage{amsthm}
\usepackage{graphicx}
\usepackage{amsfonts}
\usepackage{amssymb}
\usepackage{color}
\usepackage{xcolor}
\usepackage{booktabs}
\usepackage{tabularx}
\usepackage[mathlines]{lineno}
\usepackage{needspace}

\modulolinenumbers[1]
\newtheorem{defn}{Definition}
\newtheorem{thm}{Theorem}
\newtheorem{lem}{Lemma}
\newtheorem{prop}{Proposition}

\newtheorem{fact}{Fact}
\newtheorem{cor}{Corollary}

\usepackage{enumitem}
\usepackage{comment}
\usepackage{tikz}
\usepackage{hyperref}
\usetikzlibrary{arrows.meta, positioning, calc}
\usetikzlibrary{positioning}

\definecolor{qLight}{gray}{0.85}
\definecolor{qDark}{gray}{0.70}
\definecolor{qfigReference}{RGB}{34,100,164}
\definecolor{qfigOutside}{RGB}{197,100,24}
\definecolor{qfigKeep}{RGB}{0,119,105}

\journal{Discrete Applied Mathematics}

\begin{document}

\begin{frontmatter}

\title{Strong Testability of Full Quartet Consistency with Polynomial Query Bounds\tnoteref{label_title}}
\tnotetext[label_title]{This work is supported by the National Science and Technology Council of Taiwan under grant no.~NSTC 115-2221-E-019-045-.}

\author[CCLin]{Chuang-Chieh Lin\corref{cor}}
\ead{josephcclin@mail.ntou.edu.tw}

\affiliation[CCLin]{organization={Department of Computer Science and Engineering, National Taiwan Ocean University},
    addressline={No.2, Beining Rd., Jhongjheng Dist.}, 
    city={Keelung City},
    postcode={202301}, 
    %state={},
    country={Taiwan}}

\cortext[cor]{Corresponding author}

\begin{abstract}
We study dense property testing of full quartet systems on~$n$ taxa, 
where one resolved topology is specified for every four-set.  
We give an explicit polynomial-time adaptive one-sided tester using
$O(\varepsilon^{-1}\log(2/\varepsilon)\log(1/\delta))$ quartet queries to distinguish tree-like inputs from inputs that are $\varepsilon$-far from every tree-induced system, with soundness error at most $\delta$, for $0<\delta\leq1/3$.  The bound is independent of $n$.  The tester reconstructs trees on a small random sample and its verification taxa, then checks their individual attachment edges.  Soundness follows by comparing with a global tree fixed by the initial sample and bounding attachment collisions through a conditional symmetry argument.  Querying the complete sampled subsystem in advance also gives an explicit non-adaptive $O(\varepsilon^{-4}\log(1/\delta))$ bound.  A distance-preserving encoding into directed colored structures provides an additional qualitative proof through hereditary testing.  For ordinary testing, we prove an $\Omega(\varepsilon^{-1}\log(1/\delta))$ query lower bound for $0<\varepsilon,\delta\leq1/4$ and sufficiently large $n$.  We also analyze the stronger requirement of reconstructing the exact tree on tree-like inputs, obtaining tight bounds up to constants of $\Theta(n\log n+\varepsilon^{-1}\log(1/\delta))$ adaptively and $\Theta(n^3+\varepsilon^{-1}\log(1/\delta))$ non-adaptively.
\end{abstract}


\end{frontmatter}

%\maketitle
%\enlargethispage{6pt}
%\linenumbers

%%%%%%%%%%%%%%%%%%%%%%%%%%%%%%%%%%%%%%%%%%%%%%%%%%%%%%%%%%%%%%%%%%%
\section{Introduction}
\label{sec:intro}
%%%%%%%%%%%%%%%%%%%%%%%%%%%%%%%%%%%%%%%%%%%%%%%%%%%%%%%%%%%%%%%%%%%

Quartet topologies are a standard local representation of an unrooted phylogenetic 
tree~\cite{SempleSteel2003}: for every set of four taxa, the tree displays 
exactly one of the three resolved quartet topologies, and the collection of all 
its displayed quartets uniquely determines the tree~\cite{colonius1981tree}. 
In empirical data, however, independently inferred quartet topologies
need not be jointly realizable by any tree.  This motivates the algorithmic problem
of distinguishing a tree-like full quartet system from one that requires changing
a positive fraction of all quartet topologies.

Property testing is a central paradigm in sublinear-time algorithms.
Given oracle access to a large object, a tester seeks to distinguish
objects satisfying a property from those that are far from satisfying
it, while inspecting only a sublinear number of input 
entries~\cite{GoldreichGoldwasserRon98,Goldreich17}. This relaxation has led
to general techniques for analyzing functions, graphs, hypergraphs,
distributions, and other large combinatorial structures~\cite{Goldreich17}.

Full quartet systems fit naturally into the dense oracle model.
A full system on~$n$ taxa contains $\binom n4=\Theta(n^4)$ entries.
A query selects four taxa and returns their resolved split.
The property-testing task is to distinguish tree-like systems from systems 
requiring at least an~$\varepsilon$ fraction of all quartet entries to change, 
with error probability at most~$\delta$, using a number of queries sublinear 
in the total number of entries. 
Our main result gives an explicit one-sided tester whose query bound is polynomial 
in~$1/\varepsilon$ and independent of~$n$. 
For fixed~$\varepsilon$ and~$\delta$, the number of queried entries is therefore 
constant as~$n$ grows.

%==================================================================
\subsection{The Model and Testing Framework}
\label{subsec:model}
%==================================================================

A \emph{taxon} is a leaf label, representing a species or another
object whose relationships are described by a tree.  Let~$X$ be a known 
and fixed set of~$n\geq4$ taxa.  An \emph{unrooted phylogenetic $X$-tree}
is a finite unrooted tree with leaf set bijectively labeled by~$X$ and
with every internal vertex of degree~$3$. Trees are compared up to
isomorphism preserving the leaf labels.

\begin{defn}[Restriction of a phylogenetic tree]
\label{def:restrict-tree}
For a tree $T$ and a subset~$S$ of at least two of its leaf labels,
let~$T[S]$ be the minimal connected subgraph spanning the leaves in~$S$.
The \emph{restriction} $T|_S$ is obtained from $T[S]$ by suppressing
every degree-$2$ vertex, replacing its two incident edges by a single
edge.  Its leaf set is~$S$.  For nested subsets~$Y\subseteq S$,
restriction is transitive, so~$(T|_S)|_Y\cong T|_Y$.
\end{defn}

For any four-set $Y$, let $\mathcal Q(Y)$ be the three unordered
bipartitions of $Y$ into pairs.  For example, if~$Y=\{a,b,c,d\}$,
then $\mathcal Q(Y)=\{ab\mid cd,ac\mid bd,ad\mid bc\}$, where
$ab\mid cd$ denotes $\{\{a,b\},\{c,d\}\}$.  The notation is unchanged
by reversing a pair or interchanging the two sides.

\begin{defn}[Full quartet system]
\label{def:full_quartet_system}
A \emph{full quartet system} on~$X$ assigns to each
$Y\in\binom X4$ exactly one topology $Q(Y)\in\mathcal Q(Y)$.
\end{defn}

\begin{defn}[Displaying quartets and quartet systems]
\label{def:display}
Let $T$ be a phylogenetic~$X$-tree and~$Y\in\binom X4$.
Deleting the unique internal edge of~$T|_Y$ partitions its four
leaf labels into two pairs.  We say that~$T$ \emph{displays}
$q\in\mathcal Q(Y)$ when this bipartition equals~$q$, and denote
that uniquely displayed topology by~$Q_T(Y)$.
The tree displays a full quartet system $Q$ when $Q(Y)=Q_T(Y)$
for every~$Y\in\binom X4$.
For example, Fig.~\ref{fig:restriction-process} uses
$Y=\{a,b,c,d\}$ and displays~$ac\mid bd$.
\end{defn}

\begin{figure}[ht]
\centering
\begin{tikzpicture}[line width=.8pt,scale=.85,
 every node/.style={font=\footnotesize,text=black},every path/.style={draw=black}]
\begin{scope}
\coordinate (u) at (0,0);
\coordinate (v) at (1,0);
\coordinate (w) at (2,0);
\coordinate (z) at (3,0);
\draw (u)--(-.55,.65) node[above] {$a$};
\draw (u)--(-.55,-.65) node[below] {$c$};
\draw (u)--(v)--(w)--(z);
\draw (v)--(1,-.7) node[below] {$f$};
\draw (w)--(2,.7) node[above] {$b$};
\draw (z)--(3.55,.65) node[above] {$e$};
\draw (z)--(3.55,-.65) node[below] {$d$};
\node at (1.5,-1.60) {(i) $T$ on $X=\{a,b,c,d,e,f\}$};
\end{scope}
\draw[-{Stealth[length=2mm]}] (4.1,0)--(4.7,0);
\begin{scope}[xshift=5.2cm]
\draw (0,0)--(-.45,.65) node[above] {$a$};
\draw (0,0)--(-.45,-.65) node[below] {$c$};
\draw (0,0)--(.8,0)--(1.6,0)--(2.4,0);
\draw (1.6,0)--(1.6,.7) node[above] {$b$};
\draw (2.4,0)--(2.85,-.65) node[below] {$d$};
\fill (.8,0) circle (1.3pt);
\fill (2.4,0) circle (1.3pt);
\node at (1.2,-1.60) {(ii) $T[Y]$, $Y=\{a,b,c,d\}$};
\end{scope}
\draw[-{Stealth[length=2mm]}] (8.65,0)--(9.25,0);
\begin{scope}[xshift=9.8cm]
\draw (0,0)--(-.45,.65) node[above] {$a$};
\draw (0,0)--(-.45,-.65) node[below] {$c$};
\draw[line width=1.4pt] (0,0)--(1.2,0);
\draw (1.2,0)--(1.65,.65) node[above] {$b$};
\draw (1.2,0)--(1.65,-.65) node[below] {$d$};
\node[align=center] at (1.4,-1.60) {(iii) $T|_Y$; $Q_T(Y)=ac\mid bd$};
\end{scope}
\end{tikzpicture}
\caption{Restriction to the specific four-set $Y=\{a,b,c,d\}$.
The degree-$2$ vertices marked in (ii) are suppressed to obtain (iii).
Deleting the thicker internal edge in (iii) separates $\{a,c\}$
from~$\{b,d\}$, so the displayed quartet is $ac\mid bd$.}
\label{fig:restriction-process}
\end{figure}

\begin{defn}[Tree-likeness]
\label{def:tree-likeness}
A full quartet system $Q$ on a taxon set~$X$ is
\emph{tree-like} if it is displayed by some unrooted phylogenetic
$X$-tree $T$ in the sense of Definition~\ref{def:display}. Equivalently,
$Q=Q_T$.
\end{defn}

For a fixed taxon set $X$, let $\mathcal{P}_{\text{tree}}(X)$ denote the set of all tree-like full quartet systems on~$X$.
When $X$ is clear from context, we write simply~$\mathcal{P}_{\text{tree}}$.

\begin{defn}[Restriction of a quartet system]
\label{def:restrict-quartet-system}
Let $Q$ be a full quartet system on $X$.  For any
$S\subseteq X$ with $|S|\geq4$, the \emph{restriction} of $Q$ to $S$
is the full quartet system 
\[
Q|_S:\binom{S}{4}\longrightarrow\bigcup_{Y\in\binom S4}\mathcal Q(Y)
\quad\text{defined by}\quad Q|_S(Y):=Q(Y)\quad\text{for all }Y\in\binom S4.
\]
\end{defn}

\paragraph{Distance to tree-likeness}
For two full quartet systems $Q,Q'$ on the same $n$-taxon set
$X$, define their normalized Hamming distance by 
\[
\operatorname{dist}(Q,Q') :=
\frac{\bigl|\{Y\in\binom X4:Q(Y)\ne Q'(Y)\}\bigr|}{\binom n4}.
\]
For a property $\mathcal P$ of full quartet systems on $X$, let 
\[
\operatorname{dist}(Q,\mathcal P):=
\min_{Q^\star\in\mathcal P}\operatorname{dist}(Q,Q^\star).
\]
We say that $Q$ is \emph{$\varepsilon$-far} from $\mathcal{P}$
if $\operatorname{dist}(Q,\mathcal{P})\geq \varepsilon$.  
Thus one must modify the topology of at least $\varepsilon\binom{n}{4}$ quartets to
obtain a member of $\mathcal{P}$. 
%This notion matches the dense property-testing model for colored complete $4$-uniform hypergraphs, where 
Each modification replaces the selected topology on one four-set.

\paragraph{Quartet-query oracle}
A \emph{quartet query} to a full quartet system $Q$ on $X$
submits an unordered four-set $Y=\{a,b,c,d\}\in\binom X4$ and
receives the answer
\[
Q(Y)\in\mathcal{Q}(Y)
=\{ab\mid cd,\;ac\mid bd,\;ad\mid bc\}.
\]
Each queried four-set counts as one oracle query.
Taxon labels are known, and uniform sampling of taxa and four-sets
requires no quartet queries. Note that query complexity only counts 
oracle calls. An algorithm may cache answers to previously queried 
four-sets\footnote{Sampling and looking up stored answers still take 
computational time. The query complexity measures the number of calls 
made to the quartet oracle.}.

\begin{defn}[Property tester for tree-likeness]\label{defn:PT}
A randomized \emph{property tester} for tree-likeness has oracle
access to a full quartet system $Q$ on an $n$-taxon set and, given
$\varepsilon,\delta\in(0,1)$, makes at most
$q=q(n,\varepsilon,\delta)$ queries before outputting \textsf{accept}
or \textsf{reject}.  It has \emph{one-sided error} if it accepts every
tree-like $Q$ with probability~$1$ and rejects every
$\varepsilon$-far $Q$ with probability at least $1-\delta$.  A
two-sided tester may reject a tree-like input with probability at most
$\delta$.  A tester is \emph{adaptive} if a query may depend on earlier
answers, and is \emph{non-adaptive} otherwise.
Either tester may behave arbitrarily when the input is neither
tree-like nor $\varepsilon$-far.  Ordinary property testing has only
the two outputs \textsf{accept} and \textsf{reject}.  We say that 
tree-likeness is \emph{strongly testable} if, for 
every~$\varepsilon,\delta\in(0,1)$, there is a one-sided tester whose 
query bound~$q(\varepsilon,\delta)$ is independent of~$n$.
\end{defn}

%====================================================================
\subsection{Our Contributions}
\label{subsec:contribution}
%====================================================================

Our main contribution is an explicit tester whose query complexity is
polynomial in~$1/\varepsilon$ and independent of the size~$n$ of the taxon 
set~$X$.
\begin{enumerate}[leftmargin=*,label=(\roman*)]
\item \emph{An explicit adaptive tester for ordinary testing.}
Theorem~\ref{thm:polynomial-tester} and
Corollary~\ref{cor:polynomial-amplification} give a polynomial-time,
one-sided tester using
\[
O\!\left(\varepsilon^{-1}\log(2/\varepsilon)\log(1/\delta)\right)
\]
queries, independently of $n$, for $0<\delta\leq1/3$.
The construction learns trees on $O(\varepsilon^{-1})$ sampled taxa
and checks their one-leaf extensions.  A global comparison tree fixed
by the initial sample and a conditional bound on attachment collisions
provide the soundness argument.  Caching the complete sampled system
gives an explicit non-adaptive
$O(\varepsilon^{-4}\log(1/\delta))$ tester
(Corollary~\ref{cor:polynomial-nonadaptive}).
\item \emph{A structural explanation through hereditary testing.}
A relabeling-equivariant encoding preserves quartet distance and
places tree-likeness within the directed colored hereditary-testing
framework of Austin and Tao~\cite{AustinTao10}.  We give the required
definitions and theorem explicitly.  This supplies a qualitative
explanation of testability alongside the quantitative algorithm.
\item \emph{Lower bounds and exact reconstruction.}
For ordinary testing, every possibly adaptive two-sided tester satisfies
\[
\liminf_{n\to\infty}q(n,\varepsilon,\delta)
\geq\frac{\ln((1-\delta)/\delta)}{\ln(1/(1-\varepsilon))}.
\]
For one-sided completeness the numerator strengthens to $\ln(1/\delta)$
(Corollary~\ref{cor:sharp-testing-lb}).  For fixed~$\delta<1/2$,
this is an $\Omega(\varepsilon^{-1})$ bound as $\varepsilon\to0$,
leaving a logarithmic gap to our adaptive upper bound.
We also retain explicit reconstruct-and-verify algorithms using
$O(n\log n+\varepsilon^{-1}\log(1/\delta))$ adaptive queries and
$\binom{n-1}{3}+O(\varepsilon^{-1}\log(1/\delta))$ non-adaptive
queries.  They solve the stronger reconstruct-or-reject task,
which requires exact reconstruction on tree-like inputs.
For that task their bounds are tight up to constant factors for
sufficiently large~$n$ (Corollary~\ref{cor:reconstruct-or-reject-tight}).
\end{enumerate}

\paragraph{Remark} The $n$-dependent lower bounds concern exact reconstruction.
The sampled tester of Theorem~\ref{thm:polynomial-tester} decides
ordinary tree-likeness and queries a number of entries independent of~$n$.

\begin{table}[ht]
\caption{Query bounds for full quartet systems with error probability
$1/3$.  All upper bounds are one-sided.  Reconstruction bounds are
asymptotic for sufficiently large $n$ as a function of $\varepsilon$. Results in boldface refer to this work.}
\label{tab:comparison}
\centering\footnotesize\color{black}
\renewcommand{\arraystretch}{1.15}
\begin{tabularx}{\textwidth}{@{}p{.32\textwidth}p{.17\textwidth}p{.12\textwidth}X@{}}
\toprule
Result & Task & Model & Bound \\
\midrule
Chang--Lin--Rossmanith~\cite{CLR11}
& Ordinary testing & Non-adaptive & $O(n^3/\varepsilon)$, explicit.\\
\textbf{Theorem~\ref{thm:polynomial-tester}}
& Ordinary testing & Adaptive & $O(\varepsilon^{-1}\log(2/\varepsilon))$, explicit.\\
Corollary~\ref{cor:polynomial-nonadaptive}
& Ordinary testing & Non-adaptive & $O(\varepsilon^{-4})$, explicit.\\
\textbf{Proposition~\ref{prop:exist}}, via~\cite{AustinTao10}
& Ordinary testing & Non-adaptive & $q_{\rm her}(\varepsilon,1/3)$, qualitative.\\
\textbf{Theorem~\ref{thm:learn-audit} and Corollary~\ref{cor:reconstruct-or-reject-tight}}
& Reconstruct or reject & Adaptive & $\Theta(n\log n+\varepsilon^{-1})$.\\
\textbf{Corollaries~\ref{cor:cached-anchor} and~\ref{cor:reconstruct-or-reject-tight}}
& Reconstruct or reject & Non-adaptive & $\Theta(n^3+\varepsilon^{-1})$.\\
\textbf{Corollary~\ref{cor:sharp-testing-lb}}
& Ordinary testing & Adaptive allowed & $\Omega(\varepsilon^{-1})$ lower bound.\\
\bottomrule
\end{tabularx}
\end{table}

Table~\ref{tab:comparison} compares results in the full-system oracle
model.  The related-work discussion also covers incomplete quartet
inputs and passive noisy sampling.

%==================================================================
\subsection{Related Work}
\label{subsec:related_work}
%==================================================================

\paragraph{Compatibility and optimization}
For a possibly incomplete quartet set, deciding whether all
specified quartets are compatible with one phylogenetic tree is
\textsf{NP}-complete~\cite{S92}. By contrast, full quartet systems 
admit polynomial-time consistency recognition through the local
characterization recalled in
Proposition~\ref{prop:tree-like-quintets_fixed_taxon}~\cite{BD86,CLR11}.  The
maximum quartet consistency problem asks for a tree displaying as many
input quartets as possible, and its dual minimizes the number of
inconsistencies.  These optimization problems are
\textsf{NP}-hard~\cite{BJKLW99}. The maximization problem admits a
PTAS~\cite{JKL01}, while the minimization problem has fixed-parameter
algorithms parameterized by the number of inconsistencies~\cite{GN03,CLR10}.

\paragraph{Testing full quartet systems}
Chang, Lin, and Rossmanith~\cite{CLR11} introduced property
testing for full quartet tree-likeness.  Their non-adaptive one-sided
tester samples quintets containing a fixed taxon and uses
$O(n^3/\varepsilon)$ queries.  Their later parameterized tester handles
missing quartet topologies and runs in
$O(1.7321^kkn^3/\varepsilon)$ time when the number of missing entries
is~$k$~\cite{CLR13}.  Alon, Snir, and Yuster~\cite{AlonSnirYuster14}
showed that dense \emph{partial} quartet sets can be highly
incompatible even though all subsets of a prescribed constant size
are compatible.  Their negative result concerns a different input and
sampling model and does not contradict the testability of full systems.

\paragraph{Adaptive query reconstruction and passive quartet learning}
An unknown phylogenetic tree can be reconstructed from noiseless
adaptive triplet or quartet queries using $O(n\log n)$ queries
\cite{WuKLY08,EmamjomehZadehKempe18}.  Fixing an outgroup taxon
simulates rooted-triplet queries by quartet queries.  The deterministic
insertion search is an established ingredient of our algorithms.
The new sampled tester additionally checks one-leaf restrictions
and controls the dependence between reconstruction and verification
through an attachment-collision argument.  For exact reconstruction
on tree-like inputs, we also analyze the familiar reconstruction
followed by independent uniform verification.
Arvanitakis et al.~\cite{ArvanitakisEtAl26} learn a tree close to a
ground-truth tree from~$\Theta(n)$ uniformly sampled quartets under
random classification\footnote{Here the classification is assigning 
one of the three possible quartet topologies to a four-set.} noise, 
for fixed accuracy and noise parameters. 
In comparison, our tester provides worst-case guarantees 
for arbitrary full quartet systems, distinguishing tree-like systems 
from systems that are~$\varepsilon$-far from every tree-induced system.

Emamjomeh-Zadeh and Kempe~\cite{EmamjomehZadehKempe18}
also proved an $\Omega(n^3)$ lower bound for non-adaptive exact
reconstruction from rooted ordinal queries, and
Arvanitakis et al.~\cite{ArvanitakisEtAl26} extended the adaptive
upper bound and the non-adaptive lower bound to unrooted quartet
queries.  In Sect.~\ref{sec:LB}, we give self-contained lower
bounds for randomized reconstruction algorithms and combine them with
the property-testing lower bound to characterize the
reconstruct-or-reject task.

\paragraph{Hereditary property testing}
R\"odl and Schacht~\cite{RodlSchacht07} established testing results 
for colored uniform hypergraphs, and Austin and Tao~\cite{AustinTao10} 
developed a general theory for hereditary properties of multiple directed 
polychromatic hypergraphs.  
We apply the latter framework only after constructing a relabeling-equivariant 
quartet encoding.  
We therefore assign colors according to the relative positions of the taxa 
within each ordered 4-tuple, rather than according to any global ordering of 
the taxon labels.

\paragraph{Organization}
Section~\ref{sec:preliminaries} recalls the local structural facts.
Section~\ref{sec:polynomial-tester} gives the sampled tester,
its deterministic subroutines, and the complete soundness proof.
Section~\ref{sec:learn-audit} introduces reconstruct-and-verify algorithms.
Section~\ref{sec:reduction} explains the hereditary reduction.
Section~\ref{sec:LB} gives ordinary-testing and reconstruction lower
bounds.  Section~\ref{sec:concluding} concludes this work and states 
the future work.

\section{Preliminaries}
\label{sec:preliminaries}
%%%%%%%%%%%%%%%%%%%%%%%%%%%%%%%%%%%%%%%%%%%%%%%%%%%%%%%%%%%%%%%%%%%

\begin{thm}[Tree-likeness is hereditary]
\label{thm:treelike-hereditary}
Let $X$ be a finite set of taxa with $|X|\geq 4$, and let $Q$
be a full quartet system on $X$. If $Q$ is tree-like on $X$, then for every subset
$S\subseteq X$ with $|S|\ge 4$, the restricted quartet system
\[
Q|_S(Y):=Q(Y)\;\;\text{for every }\; Y\in\binom{S}{4}
\]
is tree-like on~$S$.    
\end{thm}
\iffalse
\begin{proof}
Assume that $Q$ is tree-like on $X$. By definition, there exists an unrooted phylogenetic
$X$-tree $T$ such that for every $Y\in\binom{X}{4}$, the quartet topology $Q(Y)$ equals
the topology displayed by~$T$ on~$Y$.

Fix any subset $S\subseteq X$ with $|S|\geq 4$. Let $T|_S$ denote the restriction of~$T$
to leaf set $S$, defined as follows. Take the minimal connected subgraph of $T$ that spans
the leaves in $S$, and then suppress all degree-$2$ vertices (i.e., iteratively replace every
maximal path whose internal vertices have degree $2$ by a single edge). The resulting tree
$T|_S$ is an unrooted phylogenetic $S$-tree.

We claim that $T|_S$ displays $Q|_S$. Let $Y\in\binom{S}{4}$ be arbitrary. Consider the
minimal connected subgraph of $T$ spanning $Y$; this subgraph is contained in the minimal
connected subgraph of $T$ spanning $S$, hence it is also contained in the restriction
construction of $T|_S$. Therefore, when we further restrict $T|_S$ to $Y$ (i.e., take the
minimal connected subgraph spanning $Y$ and suppress degree-$2$ vertices), we obtain the
same quartet tree as when we restrict $T$ directly to~$Y$. In particular, the displayed
quartet topology on $Y$ is unchanged. That is, $(T|_S)|_Y \;=\; T|_Y$ which implies that $\text{quartet}(T|_S,Y) \;=\; \text{quartet}(T,Y)$.
Since $Q(Y)$ is the quartet displayed by $T$ on~$Y$, it follows that $Q|_S(Y)=Q(Y)$ is
also the quartet displayed by $T|_S$ on $Y$. As $Y\in\binom{S}{4}$ is arbitrary, $T|_S$ displays $Q|_S$, and hence $Q|_S$ is tree-like on~$S$.
\end{proof}
\fi
\begin{proof}[Proof.]
Assume that $Q$ is tree-like on $X$. By definition, there exists an
unrooted phylogenetic $X$-tree $T$ such that, for every
$Y\in\binom{X}{4}$, the quartet topology $Q(Y)$ is displayed by~$T$.
Fix any subset $S\subseteq X$ with $|S|\geq4$. We claim that the
restricted tree $T|_S$ displays $Q|_S$.

Let $Y\in\binom{S}{4}$ be arbitrary. Since $Y\subseteq S\subseteq X$, 
restriction is transitive. Indeed, in a tree, the minimal connected 
subgraph spanning a vertex set is the union of the unique paths joining 
its vertices. Since $Y\subseteq S$, every path joining two leaves of~$Y$ 
is contained in the minimal connected subgraph spanning~$S$.  
Consequently, $T[Y]\subseteq T[S]$.  
Starting from $T[S]$, suppressing its degree-$2$ vertices to form
$T|_S$, then restricting to $Y$ and suppressing any newly created
degree-$2$ vertices, has the same effect as starting directly from
$T[Y]$ and suppressing all of its degree-$2$ vertices. Therefore,
\[
(T|_S)|_Y\cong T|_Y,
\]
where $\cong$ denotes isomorphism of leaf-labeled trees. Consequently,
$T|_S$ and $T$ display the same quartet topology on~$Y$.

Since $T$ displays $Q(Y)$ on $Y$, the tree $T|_S$ also displays
$Q(Y)$. Moreover, by the definition of the restricted quartet system,
$(Q|_S)(Y)=Q(Y)$. Thus, $T|_S$ displays $(Q|_S)(Y)$. Because
$Y\in\binom{S}{4}$ was arbitrary, $T|_S$ displays every quartet in
$Q|_S$. Hence $Q|_S$ is tree-like on~$S$.
\end{proof}

\begin{figure}[ht]
\centering
    \begin{tikzpicture}[
    thick,
    scale=0.70,
    leaf/.style={font=\itshape, inner sep=3pt, outer sep=0pt},
    label_rom/.style={font=\footnotesize, below=1.8cm}
]

% Define horizontal and vertical spacing between trees
\def\hsep{4.0}
\def\vsep{-5.0}

% --- Helper macros for different tree topologies ---

% Topology A: Asymmetric, split on Top-Left branch
% Arguments: #1=a_pos, #2=e_pos, #3=b_pos, #4=c_pos, #5=d_pos, #6=RomanLabel
\newcommand{\treeAsymTL}[6]{
    \coordinate (C1) at (-0.5, 0); \coordinate (C2) at (0.5, 0);
    \draw[line width=1.2pt] (C1) -- (C2);
    \coordinate (Split) at (-1.1, 0.6);
    \draw[line width=1.2pt] (C1) -- (Split);
    \draw[line width=1.2pt] (C1) -- (-1.1, -0.6) node[leaf, below left] {#3};
    \draw[line width=1.2pt] (Split) -- +(-0.5, 0.4) node[leaf, above left] {#1};
    \draw[line width=1.2pt] (Split) -- +(0.5, 0.4) node[leaf, above right] {#2};
    \draw[line width=1.2pt] (C2) -- +(0.8, 0.7) node[leaf, above right] {#4};
    \draw[line width=1.2pt] (C2) -- +(0.8, -0.7) node[leaf, below right] {#5};
    \node[label_rom] at (0,0) {(#6)};
}

% Topology B: Asymmetric, split on Bottom-Left branch
\newcommand{\treeAsymBL}[6]{
    \coordinate (C1) at (-0.5, 0); \coordinate (C2) at (0.5, 0);
    \draw[line width=1.2pt] (C1) -- (C2);
    \draw[line width=1.2pt] (C1) -- (-1.1, 0.6) node[leaf, above left] {#1};
    \coordinate (Split) at (-1.1, -0.6);
    \draw[line width=1.2pt] (C1) -- (Split);
    \draw[line width=1.2pt] (Split) -- +(-0.5, -0.4) node[leaf, below left] {#2};
    \draw[line width=1.2pt] (Split) -- +(0.5, -0.4) node[leaf, below right] {#3};
    \draw[line width=1.2pt] (C2) -- +(0.8, 0.7) node[leaf, above right] {#4};
    \draw[line width=1.2pt] (C2) -- +(0.8, -0.7) node[leaf, below right] {#5};
    \node[label_rom] at (0,0) {(#6)};
}

% Topology C: Asymmetric, split on Top-Right branch
\newcommand{\treeAsymTR}[6]{
    \coordinate (C1) at (-0.5, 0); \coordinate (C2) at (0.5, 0);
    \draw[line width=1.2pt] (C1) -- (C2);
    \draw[line width=1.2pt] (C1) -- (-0.8, 0.7) node[leaf, above left] {#1};
    \draw[line width=1.2pt] (C1) -- (-0.8, -0.7) node[leaf, below left] {#2};
    \coordinate (Split) at (1.1, 0.6);
    \draw[line width=1.2pt] (C2) -- (Split);
    \draw[line width=1.2pt] (Split) -- +(-0.5, 0.4) node[leaf, above left] {#3};
    \draw[line width=1.2pt] (Split) -- +(0.5, 0.4) node[leaf, above right] {#4};
    \draw[line width=1.2pt] (C2) -- (1.1, -0.6) node[leaf, below right] {#5};
    \node[label_rom] at (0,0) {(#6)};
}

% Topology D: Asymmetric, split on Bottom-Right branch
\newcommand{\treeAsymBR}[6]{
    \coordinate (C1) at (-0.5, 0); \coordinate (C2) at (0.5, 0);
    \draw[line width=1.2pt] (C1) -- (C2);
    \draw[line width=1.2pt] (C1) -- (-0.8, 0.7) node[leaf, above left] {#1};
    \draw[line width=1.2pt] (C1) -- (-0.8, -0.7) node[leaf, below left] {#2};
    \draw[line width=1.2pt] (C2) -- (1.1, 0.6) node[leaf, above right] {#3};
    \coordinate (Split) at (1.1, -0.6);
    \draw[line width=1.2pt] (C2) -- (Split);
    \draw[line width=1.2pt] (Split) -- +(-0.5, -0.4) node[leaf, below left] {#4};
    \draw[line width=1.2pt] (Split) -- +(0.5, -0.4) node[leaf, below right] {#5};
    \node[label_rom] at (0,0) {(#6)};
}

% Topology S: Symmetric
% Arguments: #1=top_mid, #2=left_top, #3=left_bot, #4=right_top, #5=right_bot, #6=RomanLabel
\newcommand{\treeSym}[6]{
    \coordinate (Center) at (0, -0.3);
    \draw[line width=1.2pt] (Center) -- (0, 0.7) node[leaf, above] {#1};
    \coordinate (JL) at (-0.8, -0.9); \coordinate (JR) at (0.8, -0.9);
    \draw[line width=1.2pt] (JL) -- (Center) -- (JR);
    \draw[line width=1.2pt] (JL) -- +(-0.6, 0.6) node[leaf, above left] {#2};
    \draw[line width=1.2pt] (JL) -- +(-0.6, -0.6) node[leaf, below left] {#3};
    \draw[line width=1.2pt] (JR) -- +(0.6, 0.6) node[leaf, above right] {#4};
    \draw[line width=1.2pt] (JR) -- +(0.6, -0.6) node[leaf, below right] {#5};
    \node[label_rom] at (0,0) {(#6)};
}

% --- Drawing the 3x5 Grid of Trees ---

% Row 1
\begin{scope}[yshift=0cm]
    \begin{scope}[xshift=0*\hsep cm] \treeAsymTL{a}{e}{b}{c}{d}{i} \end{scope}
    \begin{scope}[xshift=1*\hsep cm] \treeAsymTL{a}{c}{b}{e}{d}{ii} \end{scope}
    \begin{scope}[xshift=2*\hsep cm] \treeSym{e}{a}{b}{c}{d}{iii} \end{scope}
    \begin{scope}[xshift=3*\hsep cm] \treeAsymTR{a}{b}{e}{c}{d}{iv} \end{scope}
    \begin{scope}[xshift=4*\hsep cm] \treeAsymBR{a}{b}{c}{e}{d}{v} \end{scope}
\end{scope}

% Row 2
\begin{scope}[yshift=1*\vsep cm]
    \begin{scope}[xshift=0*\hsep cm] \treeAsymTL{a}{e}{c}{b}{d}{vi} \end{scope}
    \begin{scope}[xshift=1*\hsep cm] \treeAsymBL{a}{c}{e}{b}{d}{vii} \end{scope}
    \begin{scope}[xshift=2*\hsep cm] \treeSym{e}{a}{c}{b}{d}{viii} \end{scope}
    \begin{scope}[xshift=3*\hsep cm] \treeAsymTR{a}{c}{e}{b}{d}{ix} \end{scope}
    \begin{scope}[xshift=4*\hsep cm] \treeAsymBR{a}{c}{b}{e}{d}{x} \end{scope}
\end{scope}

% Row 3
\begin{scope}[yshift=2*\vsep cm]
    \begin{scope}[xshift=0*\hsep cm] \treeAsymTL{a}{e}{d}{b}{c}{xi} \end{scope}
    \begin{scope}[xshift=1*\hsep cm] \treeAsymBL{a}{d}{e}{b}{c}{xii} \end{scope}
    \begin{scope}[xshift=2*\hsep cm] \treeSym{e}{a}{d}{b}{c}{xiii} \end{scope}
    \begin{scope}[xshift=3*\hsep cm] \treeAsymTR{a}{d}{e}{b}{c}{xiv} \end{scope}
    \begin{scope}[xshift=4*\hsep cm] \treeAsymBR{a}{d}{b}{e}{c}{xv} \end{scope}
\end{scope}

\end{tikzpicture}
\caption{Fifteen possible topologies for a quintet $\{a,b,c,d,e\}$. 
As there are three possible topologies for the quartet $\{a,b,c,d\}$  
and five positions for inserting $e$ on each quartet, there are exactly 
fifteen possible topologies for~$\{a,b,c,d,e\}$.}
\label{fig:quintet_topologies}
\end{figure}

A quintet is a set of five taxa.  There are $15$ binary phylogenetic 
trees on a fixed labeled quintet (see Fig.~\ref{fig:quintet_topologies}).  
We call $Z\in\binom X5$ \emph{resolved with respect to $Q$} if $Q|_Z$ is 
displayed by one of these trees, and \emph{unresolved} otherwise.

\begin{prop}[Theorem~2 in~\cite{CLR11}]
\label{prop:tree-like-quintets_fixed_taxon}
Let $Q$ be a full quartet system on $X$, and fix
$\ell\in X$.  Then $Q$ is tree-like if and only if every quintet
containing $\ell$ is resolved with respect to~$Q$.
\end{prop}

\begin{cor}
\label{cor:tree-like-quintets}
Let $Q$ be a full quartet system on~$X$, where $|X|\geq5$.
Then $Q$ is tree-like if and only if every quintet in
$\binom{X}{5}$ is resolved with respect to~$Q$.
\end{cor}
\begin{proof}[Proof.]
If $Q$ is tree-like, every restriction is tree-like by
Theorem~\ref{thm:treelike-hereditary}.  Conversely, if every quintet
is resolved, then every quintet containing any fixed taxon is resolved,
so Proposition~\ref{prop:tree-like-quintets_fixed_taxon} applies.
%For the case that~$|X|=4$, every full quartet system consists of one topology
%and is displayed by a quartet tree.  
\end{proof}

\paragraph{Remark} Corollary~\ref{cor:tree-like-quintets}'s consequence is 
used below to describe the complete family of local forbidden patterns.

\begin{fact}[Fact~1 in~\cite{CLR11}]
\label{fact:quintet-distance}
The full quartet systems induced by any two distinct labeled
quintet trees disagree on at least two of their five quartet entries.
\end{fact}

There are $3^{\binom{5}{4}} = 243$ full quartet systems on a fixed
quintet.  Exactly $15$ are resolved, so the remaining $228$ are local
forbidden patterns.

%%%%%%%%%%%%%%%%%%%%%%%%%%%%%%%%%%%%%%%%%%%%%%%%%%%%%%%%%%%%%%%%%%%
\section{An Explicit Tester with Polynomial Dependence on the Inverse Distance}
\label{sec:polynomial-tester}
%%%%%%%%%%%%%%%%%%%%%%%%%%%%%%%%%%%%%%%%%%%%%%%%%%%%%%%%%%%%%%%%%%%

Theorem~\ref{thm:polynomial-tester} gives an explicit tester for
tree-likeness whose query complexity is polynomial in~$1/\varepsilon$
and independent of~$n=|X|$. 
The general idea of the tester is as follows.
We first choose a small random reference set~$S\subseteq X$.
We also independently sample uniformly random four-sets for
verification. Let~$U$ consist of~$S$ together with all taxa appearing
in these four-sets.
The tester constructs candidate trees on~$S$ and~$U$ and checks
that they agree when restricted to~$S$.
For each additional sampled taxon~$x\in U\setminus S$, it uses
quartet queries involving~$x$ and three taxa from~$S$ to choose
a candidate attachment edge in the tree on~$S$.
It then checks that attaching~$x$ to this edge gives the same
tree as restricting the candidate tree on~$U$ to~$S\cup\{x\}$.
Finally, it checks the oracle answers on the sampled verification
four-sets against the candidate tree on~$U$.
For one execution with soundness error at most~$1/3$, 
the set~$U$ contains~$O(\varepsilon^{-1})$ distinct taxa,
and every quartet query involves only taxa in~$U$. 
Fig.~\ref{fig:sample-tester-overview} summarizes the sampling and
the three checks performed by the tester.

\begin{figure}[htbp]
\centering
\color{black}
\resizebox{\linewidth}{!}{%
\begin{tikzpicture}[
  x=1cm,y=1cm,line cap=round,line join=round,
  every node/.style={font=\footnotesize,text=black},
  branch/.style={draw=black,line width=.85pt},
  refleaf/.style={circle,fill=qfigReference,inner sep=1.9pt},
  outleaf/.style={rectangle,fill=qfigOutside,inner sep=1.9pt},
  inner/.style={circle,fill=black,inner sep=1.25pt},
  note/.style={font=\footnotesize,align=center}
]
\path[use as bounding box] (0,-4.55) rectangle (14.4,3.25);
\node[note] at (7.2,3.0)
  {Choose $S$ and verification four-sets $Y_1,Y_2,\ldots,Y_m$;
   set $U=S\cup\bigcup_{i=1}^mY_i$.};

% Two independent deterministic learning runs.
\node[note] at (2.7,2.15) {(a) $B=\operatorname{\texttt{Learn}}(S)$};
\coordinate (ba) at (.65,.85);
\coordinate (bb) at (.65,-.75);
\coordinate (bu) at (1.75,.05);
\coordinate (bv) at (3.65,.05);
\coordinate (bc) at (4.75,.85);
\coordinate (bd) at (4.75,-.75);
\draw[branch] (ba)--(bu)--(bv)--(bc);
\draw[branch] (bb)--(bu) (bv)--(bd);
\foreach \p in {ba,bb,bc,bd} \node[refleaf] at (\p) {};
\foreach \p in {bu,bv} \node[inner] at (\p) {};
\node[above left=2pt,text=qfigReference] at (ba) {$a$};
\node[below left=2pt,text=qfigReference] at (bb) {$b$};
\node[above right=2pt,text=qfigReference] at (bc) {$c$};
\node[below right=2pt,text=qfigReference] at (bd) {$d$};

\node[note] at (10.8,2.15) {(b) $T=\operatorname{\texttt{Learn}}(U)$};
\coordinate (ta) at (8.3,.85);
\coordinate (tb) at (8.3,-.75);
\coordinate (tu) at (9.4,.05);
\coordinate (tw) at (8.85,.45);
\coordinate (tz) at (10.65,.05);
\coordinate (tv) at (11.9,.05);
\coordinate (tc) at (13.0,.85);
\coordinate (td) at (13.0,-.75);
\coordinate (tx) at (9.15,1.35);
\coordinate (ty) at (10.65,-.95);
\draw[branch] (ta)--(tw)--(tu)--(tz)--(tv)--(tc);
\draw[branch] (tb)--(tu) (tv)--(td);
\draw[branch,draw=qfigOutside] (tw)--(tx) (tz)--(ty);
\foreach \p in {ta,tb,tc,td} \node[refleaf] at (\p) {};
\foreach \p in {tx,ty} \node[outleaf] at (\p) {};
\foreach \p in {tu,tw,tz,tv} \node[inner] at (\p) {};
\node[above left=2pt,text=qfigReference] at (ta) {$a$};
\node[below left=2pt,text=qfigReference] at (tb) {$b$};
\node[above right=2pt,text=qfigReference] at (tc) {$c$};
\node[below right=2pt,text=qfigReference] at (td) {$d$};
\node[above=3pt,text=qfigOutside] at (tx) {$x$};
\node[below=3pt,text=qfigOutside] at (ty) {$y$};

% Color and shape both distinguish reference and additional taxa.
\node[refleaf] at (4.1,-1.6) {};
\node[anchor=west] at (4.28,-1.6) {taxa in $S$};
\node[outleaf] at (7.6,-1.6) {};
\node[anchor=west] at (7.78,-1.6) {taxa in $U\setminus S$};
\draw[draw=black!25,line width=.4pt] (.3,-1.98)--(14.1,-1.98);

% These are exactly the three rejection checks in SampleTest.
\node[anchor=west] at (.55,-2.43) {1. Reference restriction};
\node[anchor=west] at (5.75,-2.43) {$T|_S\stackrel{?}{=}B$};
\node[anchor=west] at (.55,-3.04) {2. Individual attachments};
\node[anchor=west] at (5.75,-3.04)
  {$T|_{S\cup\{x\}}\stackrel{?}{=}B+_{e_S(x)}x
   \quad\text{for every }x\in U\setminus S$};
\node[anchor=west] at (.55,-3.65) {3. Verification quartets};
\node[anchor=west] at (5.75,-3.65)
  {$Q(Y_i)\stackrel{?}{=}Q_T(Y_i)\quad\text{for every }i$};
\node[note] at (7.2,-4.34)
  {$|U|=O(\varepsilon^{-1})$; every quartet query uses only taxa in $U$.};
\end{tikzpicture}%
}
\caption{Overview of the sampled tester.
The trees $B$ and $T$ are obtained in separate deterministic learning runs.
The displayed leaf sets are schematic.
Blue circles denote reference taxa and orange squares additional sampled taxa.
The tester rejects if any of the three checks fails and otherwise accepts.
In the example, $S=\{a,b,c,d\}$, $U=S\cup\{x,y\}$, and
$Q_T(\{a,b,x,y\})=ax\mid by$.}
\label{fig:sample-tester-overview}
\end{figure}

\begin{thm}[Explicit adaptive tester]
\label{thm:polynomial-tester}
For every $n\geq4$ and $0<\varepsilon<1$, full quartet tree-likeness
has a polynomial-time adaptive one-sided tester using at most
\[
296m\left\lceil\lg(296m)\right\rceil+m,
\]
quartet queries\footnote{$\lg(x):=\log_2(x)$ denotes the binary logarithm.}, 
where~$m=\left\lceil (\ln6)/\varepsilon\right\rceil$.  
The tester accepts every tree-like input with
probability~$1$ and rejects every $\varepsilon$-far input with
probability at least $2/3$.  In particular, its query complexity is
$O(\varepsilon^{-1}\log(2/\varepsilon))$, independently of~$n$.
\end{thm}

The remainder of this section describes the tester and proves
Theorem~\ref{thm:polynomial-tester}.

%==================================================================
\subsection{Deterministic insertion and learning}
%==================================================================

Let $B$ be a phylogenetic tree with leaf set~$A$, vertex set~$V(B)$ and 
edge set~$E(B)$, where $h=|A|\geq3$, and let $x\notin A$.  
Write~$B+_e x$ for the tree obtained by subdividing edge $e\in E(B)$ 
and attaching a new leaf labeled~$x$ to the subdivision vertex.  
Every phylogenetic tree on~$A\cup\{x\}$ whose restriction
to~$A$ is~$B$ has the form~$B+_e x$ for a unique edge~$e\in E(B)$.
Indeed, deleting~$x$ and suppressing its resulting degree-$2$
neighbor identifies precisely that edge. Fig.~\ref{fig:attachment-and-search}(a) illustrates this operation. 

\begin{figure}[htbp]
\centering
\color{black}
\resizebox{\linewidth}{!}{%
\begin{tikzpicture}[
  x=1cm,y=1cm,line cap=round,line join=round,
  every node/.style={font=\footnotesize,text=black},
  branch/.style={draw=black,line width=.85pt},
  refleaf/.style={circle,fill=qfigReference,inner sep=1.75pt},
  outleaf/.style={rectangle,fill=qfigOutside,inner sep=1.9pt},
  inner/.style={circle,fill=black,inner sep=1.2pt},
  note/.style={font=\footnotesize,align=center}
]
\path[use as bounding box] (-.3,-2.65) rectangle (14.4,2.9);
\node[note] at (3.0,2.65) {(a) Subdivide $e$ and attach $x$};
\node[note] at (11.0,2.65) {(b) A quartet selects a direction};
\draw[draw=black!20,line width=.4pt] (7.0,-1.6)--(7.0,2.4);

% Before insertion: e is the internal edge of a four-leaf tree.
\coordinate (a0) at (0,1.05);
\coordinate (b0) at (0,-.15);
\coordinate (u0) at (.7,.45);
\coordinate (v0) at (1.6,.45);
\coordinate (c0) at (2.3,1.05);
\coordinate (d0) at (2.3,-.15);
\draw[branch] (a0)--(u0)--(b0) (c0)--(v0)--(d0);
\draw[branch,draw=qfigOutside,line width=1.6pt] (u0)--(v0);
\foreach \p in {a0,b0,c0,d0} \node[refleaf] at (\p) {};
\foreach \p in {u0,v0} \node[inner] at (\p) {};
\node[above left=1pt,text=qfigReference] at (a0) {$a$};
\node[below left=1pt,text=qfigReference] at (b0) {$b$};
\node[above right=1pt,text=qfigReference] at (c0) {$c$};
\node[below right=1pt,text=qfigReference] at (d0) {$d$};
\node[above=3pt,text=qfigOutside] at (1.15,.45) {$e$};
\node[note] at (1.15,-.8) {$B$};
\draw[-{Stealth[length=2.4mm]},draw=black,line width=.85pt]
  (2.68,.45)--(3.24,.45);

% After insertion: the orange outlined vertex subdivides e.
\coordinate (a1) at (3.75,1.05);
\coordinate (b1) at (3.75,-.15);
\coordinate (u1) at (4.45,.45);
\coordinate (w1) at (4.9,.45);
\coordinate (v1) at (5.35,.45);
\coordinate (c1) at (6.05,1.05);
\coordinate (d1) at (6.05,-.15);
\coordinate (x1) at (4.9,1.65);
\draw[branch] (a1)--(u1)--(b1) (c1)--(v1)--(d1);
\draw[branch,draw=qfigOutside,line width=1.6pt]
  (u1)--(w1)--(v1);
\draw[branch,draw=qfigOutside] (w1)--(x1);
\foreach \p in {a1,b1,c1,d1} \node[refleaf] at (\p) {};
\foreach \p in {u1,v1} \node[inner] at (\p) {};
\node[circle,draw=qfigOutside,fill=white,line width=1pt,inner sep=2.1pt]
  at (w1) {};
\node[outleaf] at (x1) {};
\node[above left=1pt,text=qfigReference] at (a1) {$a$};
\node[below left=1pt,text=qfigReference] at (b1) {$b$};
\node[above right=1pt,text=qfigReference] at (c1) {$c$};
\node[below right=1pt,text=qfigReference] at (d1) {$d$};
\node[above=3pt,text=qfigOutside] at (x1) {$x$};
\node[note] at (4.9,-.8) {$B+_e x$};
\node[note] at (3.0,-1.7)
  {Delete $x$ and suppress its neighbor\\to recover $e$.};

% A query at v: only the green direction survives the answer xa|bc.
\coordinate (sv) at (11.0,.4);
\coordinate (sp) at (9.6,.4);
\coordinate (sa) at (8.55,1.2);
\coordinate (sa2) at (8.55,-.4);
\coordinate (sq) at (12.0,1.25);
\coordinate (sb) at (13.25,1.7);
\coordinate (sb2) at (12.0,2.0);
\coordinate (sr) at (12.0,-.45);
\coordinate (sc) at (13.25,-.9);
\coordinate (sc2) at (12.0,-1.2);
\fill[qfigKeep!8]
  (11.15,.4)--(9.65,1.55)--(8.15,1.55)--(8.15,-.75)--(9.65,-.75)--cycle;
\draw[branch,draw=black!55] (sv)--(sq)--(sb) (sq)--(sb2);
\draw[branch,draw=black!55] (sv)--(sr)--(sc) (sr)--(sc2);
\draw[draw=qfigKeep,line width=1.8pt] (sv)--(sp)--(sa) (sp)--(sa2);
\foreach \p in {sa,sa2,sb,sb2,sc,sc2} \node[refleaf] at (\p) {};
\foreach \p in {sp,sq,sr} \node[inner] at (\p) {};
\node[circle,draw=black,fill=white,line width=.85pt,inner sep=2.3pt]
  at (sv) {};
\node[above left=2pt] at (sv) {$v$};
\node[above left=1pt,text=qfigReference] at (sa) {$a$};
\node[above right=1pt,text=qfigReference] at (sb) {$b$};
\node[below right=1pt,text=qfigReference] at (sc) {$c$};
\node[note,text=qfigKeep] at (9.15,-1.1) {direction toward $a$};
\node[note] at (11.0,-2.02)
  {$Q(\{x,a,b,c\})=xa\mid bc$\\[3pt]
   Retain the edges in the direction toward~$a$.};
\end{tikzpicture}%
}
\caption{One-leaf attachment and a quartet-guided search step.
(a) The outlined orange vertex subdivides edge~$e$ of~$B$;
attaching~$x$ there gives~$B+_e x$.
(b) Representatives~$a,b,c$ are chosen one from each of the
three components of~$B-v$.
The answer~$xa\mid bc$ selects the highlighted direction
toward~$a$, consisting of all edges in the component
of~$B-v$ containing~$a$, together with the edge joining
that component to~$v$. 
The set of candidate attachment edges is initially~$E(B)$.
After each query, the search retains only the current
candidate edges in the direction selected by the answer.}
\label{fig:attachment-and-search}
\end{figure}

The following search is a direct quartet implementation of the
deterministic InsertionSort-like method discussed by
Emamjomeh-Zadeh and Kempe~\cite[Sect.~4]{EmamjomehZadehKempe18}.
We specify its behavior for every full quartet oracle, including
answers that are mutually inconsistent.  Fix a linear order on~$X$
and deterministic representations and tie rules for all trees. 
Fig.~\ref{fig:attachment-and-search}(b) illustrates how a
quartet answer guides the update of the candidate edge set.

\paragraph{\texorpdfstring{Procedure $\operatorname{\texttt{Locate}}(B,x)$}{Procedure Locate}:}
\label{alg:locate}
\begin{enumerate}[leftmargin=*,label=\arabic*.]
\item Set $C=E(B)$ and choose a fixed default edge $e_0\in E(B)$.
\item While $|C|>1$, let $B[C]$ be the subtree consisting of the
edges of $C$ and their endpoints.  Choose a vertex $v$ such that
each component of $B[C]-v$ has at most $|V(B[C])|/2$ vertices.
\item The three components of $B-v$ each contain a leaf of $B$.
Choose representatives $a,b,c$, one from each component, using the
fixed tie rules, and query $Q(\{x,a,b,c\})$.
\item If the answer is $xa\mid bc$, retain the edges of $C$
in the direction toward~$a$. The answers $xb\mid ac$ and
$xc\mid ab$ select the directions toward~$b$ and~$c$,
respectively. For each $t\in\{a,b,c\}$, the direction
toward~$t$ consists of all edges in the component of~$B-v$
containing~$t$, together with the edge joining that
component to~$v$.
If the retained set is empty, return~$e_0$.
\item When one edge remains, return that edge.
\end{enumerate}

\begin{lem}[Insertion guarantee]
\label{lem:insertion-search}
The procedure $\operatorname{\texttt{Locate}}(B,x)$ always returns an edge
of $B$ after at most
\[
d(h):=\left\lceil\lg(2h-3)\right\rceil
\]
queries, all on taxa in $A\cup\{x\}$.  If the answers on this set
are induced by $B+_e x$, it returns $e$.
\end{lem}
\begin{proof}[Proof.]
Initially $C$ is connected.  Intersecting a connected edge set with
one direction from a vertex preserves connectedness whenever the
intersection is nonempty.  Thus every executed iteration has a
well-defined candidate subtree~$B[C]$.

Every finite tree has a vertex whose deletion leaves components
of size at most half the vertex count.  One obtains such a vertex
by walking toward a component containing more than half the
vertices until no such component remains.  When $q=|C|>1$, the
tree $B[C]$ has $q+1\geq3$ vertices, and such a vertex~$v$ has
degree at least two in $B[C]$.  It is therefore an internal,
degree-$3$ vertex of~$B$.

For each direction from $v$, the number of candidate edges in that
direction equals the number of vertices in the corresponding
component of $B[C]-v$.  Consequently, a nonempty update retains at
most $\lfloor(q+1)/2\rfloor=\lceil q/2\rceil$ edges.  Since a phylogenetic 
tree on~$h$ leaves has~$2h-3$ edges, at most~$d(h)=\left\lceil\lg(2h-3)\right\rceil$ 
iterations are required.  The default rule also terminates within this budget.

\iffalse
If the oracle is induced by $B+_e x$ and~$e$ lies in the direction
of~$a$, the displayed quartet is $xa\mid bc$.  The analogous
statements hold in the other two directions.  Inductively, the
true edge remains in~$C$ after every query, so the default rule is
never used and the returned edge is~$e$.
\fi
Suppose the oracle is induced by $B+_e x$.
In each iteration, let $v$ be the chosen search vertex and
let $a,b,c$ be the representative taxa chosen in Step~3,
one from each of the three components of~$B-v$.
If~$e$ lies in the direction toward~$a$, the displayed
quartet on $\{x,a,b,c\}$ is~$xa\mid bc$.
The analogous statements hold for the directions toward~$b$ and~$c$.
Since $e\in C$ initially, induction shows that the true edge
remains in~$C$ after every query.
Thus the retained set is never empty, the default rule is
never used, and the unique remaining edge is~$e$.
\end{proof}

\paragraph{\texorpdfstring{Procedure $\operatorname{\texttt{Learn}}(A)$}{Procedure Learn}:}
For $|A|=k\geq 3$, start with the unique three-leaf tree on the
first three taxa of~$A$ in the fixed order.  Insert each remaining
taxon~$x$ by calling $\operatorname{\texttt{Locate}}(B,x)$ on the current
tree~$B$ and replacing~$B$ by the resulting extension.  
This procedure always returns a binary phylogenetic tree on~$A$.

\begin{lem}[Deterministic learning on arbitrary inputs]
\label{lem:small-set-learning}
The procedure $\operatorname{\texttt{Learn}}(A)$ uses only quartets contained
in $A$ and makes at most~$k\lceil\lg(2k)\rceil$ queries.  
If~$Q|_A$ is tree-like, it returns the unique inducing tree.  
For fixed~$Q$, its output depends only on the set~$A$.
\end{lem}

\begin{proof}[Proof.]
There are totally $k-3$ insertions.  Each uses at most
$\lceil\lg(2k)\rceil$ queries by Lemma~\ref{lem:insertion-search}.  
On a tree-like input, induction on the insertion order shows that 
each current tree is the restriction of the inducing tree to the 
inserted taxa.  
The next insertion is therefore correct by that lemma.  
Fixed leaf orders, tie rules, and default rules make the output a 
deterministic function of~$A$ and~$Q$.
\end{proof}

%\Needspace{22\baselineskip}
%=======================================================================
\subsection{The tester}
%=======================================================================

Using the procedures $\operatorname{\texttt{Locate}}$ and $\operatorname{\texttt{Learn}}$ 
defined above, we now introduce our tester $\operatorname{\texttt{SampleTest}}$. 

\paragraph{\texorpdfstring{Algorithm $\operatorname{\texttt{SampleTest}}(Q,\varepsilon)$}{Algorithm \texttt{SampleTest}}:}
\label{alg:sample-tester}
Set
\begin{equation}
m=\left\lceil\frac{\ln6}{\varepsilon}\right\rceil,
\qquad s=\min\{n,144m\}.
\label{eq:sample-parameters}
\end{equation}
\begin{enumerate}[leftmargin=*,label=\arabic*.]
\item Choose $S$ uniformly from $\binom Xs$.
Independently choose $Y_1,Y_2,\ldots,Y_m$ uniformly from $\binom X4$,
with replacement, and set $U=S\cup\bigcup_{i=1}^mY_i$.
\item Compute $B=\operatorname{\texttt{Learn}}(S)$ and
$T=\operatorname{\texttt{Learn}}(U)$. %in fresh deterministic runs.
Reject if $T|_S\ne B$.
\item For every $x\in U\setminus S$, run
$e_S(x)=\operatorname{\texttt{Locate}}(B,x)$, always on the same tree $B$.
Reject if
\begin{equation}
T|_{S\cup\{x\}}\ne B+_{e_S(x)}x.
\label{eq:single-leaf-check}
\end{equation}
\item Query each $Y_i$ and reject if $Q(Y_i)\ne Q_T(Y_i)$ for any
$i$.  Otherwise accept.
\end{enumerate}
Tree comparisons preserve all taxon labels.  Repeated oracle
queries may be answered from a cache.  Each learning run follows
its fixed rules, so caching changes only the number of repeated
queries.  In particular, the tree $T$ depends on the sampled data
only through the set $U$.

%=======================================================================
\subsection{Attachment edges and their collisions}
%=======================================================================

If $R$ is a phylogenetic tree on a set containing~$S\cup\{x\}$, with
$x\notin S$ and~$|S|\geq3$, let
$\alpha_S^R(x)$ be the unique edge of $R|_S$ such that
\[
R|_{S\cup\{x\}}=(R|_S)+_{\alpha_S^R(x)}x.
\]
We say that two outside taxa~$x,y$ \emph{collide relative to~$S$ in~$R$} 
when $\alpha_S^R(x)=\alpha_S^R(y)$.  
The attachment edges are determined separately 
from~$R|_{S\cup\{x\}}$ and~$R|_{S\cup\{y\}}$, with all other
outside taxa removed in each restriction.
Thus the definition applies even when $R$ contains
additional taxa outside~$S\cup\{x,y\}$.

\begin{lem}[Number of colliding pairs]
\label{lem:collision-count}
Let $R$ be a phylogenetic tree with leaf set $L$, where $|L|=s+2$
and $s\geq4$.  At most $2(s+2)$ unordered pairs
$\{x,y\}\in\binom L2$ satisfy
\[
\alpha_{L\setminus\{x,y\}}^R(x)
=\alpha_{L\setminus\{x,y\}}^R(y).
\]
\end{lem}

\begin{proof}[Proof.]
Let $u,v$ be the internal neighbors of $x,y$, respectively.
The pair collides exactly when $u=v$ or $u,v$ are adjacent.
If $u=v$ (see Fig.~\ref{fig:attachment-collisions}(a)), 
the two leaves form a cherry, which is a pair of leaves
adjacent to the same internal vertex, attached to the remaining
tree through one point.  Deleting either leaf and then the other
identifies the same edge.  If $u,v$ are adjacent (see 
Fig.~\ref{fig:attachment-collisions}(b)), deleting both
leaves suppresses these two vertices along a single edge, again
giving the same individual attachment edge.  If $u,v$ are distinct
and nonadjacent (see Fig.~\ref{fig:attachment-collisions}(c)), 
their connecting path contains an intermediate
vertex whose third branch contains a remaining leaf.  This vertex
survives restriction and separates the two attachment edges.

Fix $x$ and let $\ell$ be the number of leaf neighbors of~$u$.
Since $R$ has at least six leaves, $\ell\in\{1,2\}$.  At most
$\ell-1$ other leaves share neighbor~$u$ (i.e., such leaves form 
cherries with~$x$).  There are $3-\ell$
internal neighbors of~$u$, each with at most two leaf neighbors.
Hence $x$ belongs to at most
$(\ell-1)+2(3-\ell)=5-\ell\leq4$ colliding pairs.
Let~$\mathcal{C}$ be the set of unordered colliding pairs, and
let~$c(x)$ denote the number of other leaves that collide with~$x$.
Each pair~$\{x,y\}\in\mathcal{C}$ is counted twice in
$\sum_{x\in L}c(x)$: once in~$c(x)$ and once in~$c(y)$.
Since~$c(x)\le 4$ for every~$x\in L$ and~$|L|=s+2$, we obtain
\[
|\mathcal{C}|
=\frac{1}{2}\sum_{x\in L}c(x)
\le \frac{1}{2}\cdot 4(s+2)
=2(s+2).
\]
\end{proof}

%=========================================================
% Compact illustration for Lemma 3 (number of colliding pairs).
% Required preamble packages: \usepackage{tikz,graphicx}
% Insert with: \input{attachment_collisions_tikz.tex}
% No additional TikZ libraries or custom manuscript macros are required.
% The drawing is capped at 12 cm, and shrinks to fit a narrower line.
% The caption is blue to mark this addition. For a clean version, replace
% \textcolor{blue}{...} in the caption by its contents.
% Every shaded triangle represents a nonempty subtree with leaves in S.
% The drawn trees are unrooted: all actual edges are undirected.
\begin{figure}[t]
\centering
\begingroup
\definecolor{collisionX}{RGB}{36,95,192}
\definecolor{collisionY}{RGB}{194,106,22}
\definecolor{collisionSame}{RGB}{121,80,166}
\def\collisionFigureWidth{0.98\linewidth}%
\ifdim0.98\linewidth>12cm\relax
  \def\collisionFigureWidth{12cm}%
\fi
\resizebox{\collisionFigureWidth}{!}{%
\begin{tikzpicture}[
  x=1cm,y=1cm,font=\small,
  collisionEdge/.style={draw=black,line width=.65pt,line cap=round},
  collisionCommon/.style={collisionEdge,draw=collisionSame,line width=1.15pt},
  collisionLeft/.style={collisionEdge,draw=collisionX,line width=1.15pt},
  collisionRight/.style={collisionEdge,draw=collisionY,line width=1.15pt},
  collisionSubtree/.style={draw=black!65,fill=black!7,line width=.55pt,line join=round},
  collisionDot/.style={circle,draw=none,fill=black,inner sep=0pt,minimum size=2.8pt},
  collisionLabel/.style={text=black,inner sep=1pt},
  collisionLeafX/.style={collisionLabel,text=collisionX},
  collisionLeafY/.style={collisionLabel,text=collisionY},
  collisionPanel/.style={collisionLabel,font=\small},
  collisionResult/.style={collisionLabel,font=\footnotesize}
]
\path[use as bounding box] (-1.75,-1.55) rectangle (9.35,1.99);

% (a) x and y share the neighbor u=v.  On restriction to S, u is
% pruned and w is suppressed, so the purple path becomes one edge.
\begin{scope}
  \node[collisionPanel] at (0,1.77) {(a) $u=v$};
  \coordinate (cw) at (0,0);
  \coordinate (cu) at (0,.62);
  \coordinate (cx) at (-.46,1.18);
  \coordinate (cy) at (.46,1.18);
  \coordinate (cl) at (-1.10,-.20);
  \coordinate (cr) at (1.10,-.20);
  \path[collisionSubtree] (cl)--(-1.50,-.96)--(-.70,-.96)--cycle;
  \path[collisionSubtree] (cr)--(.70,-.96)--(1.50,-.96)--cycle;
  \draw[collisionCommon] (cl)--(cw)--(cr);
  \draw[collisionEdge] (cw)--(cu);
  \draw[collisionLeft] (cu)--(cx);
  \draw[collisionRight] (cu)--(cy);
  \node[collisionDot] at (cw) {};
  \node[collisionDot] at (cu) {};
  \node[collisionLabel,below=2pt] at (cw) {$w$};
  \node[collisionLabel,right=3pt] at (cu) {$u=v$};
  \node[collisionLeafX,above=1pt] at (cx) {$x$};
  \node[collisionLeafY,above=1pt] at (cy) {$y$};
  \node[collisionResult] at (0,-1.29) {same attachment edge};
\end{scope}

% (b) The neighbors u and v are adjacent.  Both are suppressed when
% x,y are removed, so the whole purple path becomes one edge.
\begin{scope}[xshift=3.8cm]
  \node[collisionPanel] at (0,1.77) {(b) $u,v$ adjacent};
  \coordinate (cu) at (-.48,.20);
  \coordinate (cv) at (.48,.20);
  \coordinate (cx) at (-.48,1.18);
  \coordinate (cy) at (.48,1.18);
  \coordinate (cl) at (-1.10,-.20);
  \coordinate (cr) at (1.10,-.20);
  \path[collisionSubtree] (cl)--(-1.50,-.96)--(-.70,-.96)--cycle;
  \path[collisionSubtree] (cr)--(.70,-.96)--(1.50,-.96)--cycle;
  \draw[collisionCommon] (cl)--(cu)--(cv)--(cr);
  \draw[collisionLeft] (cu)--(cx);
  \draw[collisionRight] (cv)--(cy);
  \node[collisionDot] at (cu) {};
  \node[collisionDot] at (cv) {};
  \node[collisionLabel,below=2pt] at (cu) {$u$};
  \node[collisionLabel,below=2pt] at (cv) {$v$};
  \node[collisionLeafX,above=1pt] at (cx) {$x$};
  \node[collisionLeafY,above=1pt] at (cy) {$y$};
  \node[collisionResult] at (0,-1.29) {same attachment edge};
\end{scope}

% (c) This is the local configuration at an intermediate vertex w.
% Any longer paths between w and u or v are compressed schematically;
% every nonempty S-branch at an omitted vertex is also omitted.
% The third S-branch at w guarantees that w survives restriction.
\begin{scope}[xshift=7.6cm]
  \node[collisionPanel] at (0,1.77) {(c) $u,v$ nonadjacent};
  \coordinate (cu) at (-.65,.10);
  \coordinate (cw) at (0,.10);
  \coordinate (cv) at (.65,.10);
  \coordinate (cx) at (-.65,1.18);
  \coordinate (cy) at (.65,1.18);
  \coordinate (cl) at (-1.27,-.20);
  \coordinate (cr) at (1.27,-.20);
  \coordinate (cb) at (0,-.32);
  \path[collisionSubtree] (cl)--(-1.62,-.96)--(-.92,-.96)--cycle;
  \path[collisionSubtree] (cr)--(.92,-.96)--(1.62,-.96)--cycle;
  \path[collisionSubtree] (cb)--(-.35,-1.02)--(.35,-1.02)--cycle;
  \draw[collisionLeft] (cl)--(cu)--(cw);
  \draw[collisionRight] (cw)--(cv)--(cr);
  \draw[collisionEdge] (cw)--(cb);
  \draw[collisionLeft] (cu)--(cx);
  \draw[collisionRight] (cv)--(cy);
  \node[collisionDot] at (cu) {};
  \node[collisionDot] at (cv) {};
  \node[circle,draw=black,fill=white,line width=.7pt,
        inner sep=0pt,minimum size=4.8pt] at (cw) {};
  \node[collisionLabel,below=2pt] at (cu) {$u$};
  \node[collisionLabel,above=2pt] at (cw) {$w$};
  \node[collisionLabel,below=2pt] at (cv) {$v$};
  \node[collisionLeafX,above=1pt] at (cx) {$x$};
  \node[collisionLeafY,above=1pt] at (cy) {$y$};
  \node[collisionResult] at (0,-1.29) {distinct attachment edges};
\end{scope}
\end{tikzpicture}%
}
\endgroup
\caption{Attachment collisions.
Triangles represent nonempty subtrees with leaves in
$S=L\setminus\{x,y\}$.
In (a),(b), both leaves attach individually to the same edge of $R|_S$;
in (c), $w$ survives restriction and separates the attachment edges.}
\label{fig:attachment-collisions}
\end{figure}

%=========================================================

\begin{lem}[Distinct attachments determine the joint restriction]
\label{lem:distinct-attachments}
Let $R,R'$ be phylogenetic trees whose leaf sets contain~$S\cup W$,
where $S\cap W=\varnothing$ and $|S|\geq3$.  Suppose
$R|_S=R'|_S=B$ and
$R|_{S\cup\{x\}}=R'|_{S\cup\{x\}}$ for every $x\in W$.
If the edges $\alpha_S^R(x)$, for $x\in W$, are pairwise
distinct, then $R|_{S\cup W}=R'|_{S\cup W}$.
\end{lem}

\iffalse
\begin{proof}[Proof.]
Consider the minimal subtree spanning~$S$ before suppressing its
degree-$2$ vertices.  Any connected component containing outside
leaves attaches to this subtree at a single vertex.  That vertex
lies along the path corresponding to one edge of~$B$, because a
vertex already branching in three directions toward~$S$ has no
additional incident edge in a phylogenetic tree.  All leaves of that
component therefore have the same individual attachment edge.
Outside leaves attached at different points on the same suppressed
path also have that same edge.  Consequently, pairwise distinct
attachment edges force the restriction to~$S\cup W$ to consist
of~$B$ with each of those edges subdivided once and its assigned
leaf attached there.  This description is identical for~$R$ 
and~$R'$ and uniquely specifies the joint restriction.
\end{proof}
\fi
\begin{proof}[Proof.]
By transitivity of restriction, we may first replace~$R$ and~$R'$
by their restrictions to~$S\cup W$. Thus every leaf outside~$S$
belongs to~$W$. For each~$x\in W$, the equal one-leaf restrictions
specify a common attachment edge
$e_x=\alpha_S^R(x)=\alpha_S^{R'}(x)$ in~$B$.
These edges are pairwise distinct by assumption.

Let~$C$ be the minimal subtree of~$R$ spanning~$S$, before
suppressing its degree-$2$ vertices. Each connected component
outside~$C$ is joined to~$C$ by exactly one edge, since~$R$ is a
tree. Its attachment vertex has degree~$2$ in~$C$. Hence, it cannot be
a leaf of~$S$, and a degree-$3$ vertex of~$C$ already uses all
three incident edges in~$R$. Hence the attachment vertex lies
on a path that becomes one edge of~$B$.
All leaves in that component have this same individual attachment
edge. Leaves in different components attached at different points
along the same path also have that same attachment edge.

Pairwise distinctness therefore allows at most one leaf of~$W$
along each such path. Every degree-$2$ vertex of~$C$ has an
outside branch containing a leaf of~$W$, so each used path has
exactly one attachment vertex. Moreover, an outside component
containing only one leaf consists of that leaf alone: any
internal vertex would branch and lead to at least two outside
leaves, since every internal vertex of~$R$ has degree~$3$.
Consequently,~$R$ is obtained from~$B$ by subdividing each
edge~$e_x$ once and attaching~$x$ there.
The same argument applies to~$R'$, using exactly the same
edges~$e_x$. This uniquely determines both joint restrictions,
so $R|_{S\cup W}=R'|_{S\cup W}$.
\end{proof}

%===================================================================
\subsection{A comparison tree fixed by the initial sample}
%===================================================================

Fix $S$ and the input~$Q$, and let~$B=\operatorname{\texttt{Learn}}(S)$.
For every $x\in X\setminus S$, define
$e_S(x)=\operatorname{\texttt{Locate}}(B,x)$ by the same deterministic
procedure as in the tester.  Construct a tree~$H_S$ on~$X$ by
attaching all outside taxa to their assigned edges of~$B$.
For each edge with several assigned taxa, replace that edge by
a path with one new subdivision vertex for each assigned taxon
and attach the taxa in the fixed order along that path.  Fix the
orientation of each such path deterministically.  This construction
gives a phylogenetic tree with
\begin{equation}
H_S|_S=B,
\qquad H_S|_{S\cup\{x\}}=B+_{e_S(x)}x.
\label{eq:comparison-tree}
\end{equation}
The values for taxa outside~$U$ and the tree~$H_S$ are used only
in the proof.  The tester computes~$e_S(x)$ for~$x\in U\setminus S$.
For fixed~$Q$, the entire tree~$H_S$ is a function of~$S$ alone.

Write~$Q_{H_S}$ for the full quartet system induced by~$H_S$.
Thus, for each~$Y\in\binom X4$, $Q_{H_S}(Y)$ is the quartet
topology displayed by~$H_S$ on~$Y$.
Let $\mathsf{Miss}$ be the event that
$Q(Y_i)=Q_{H_S}(Y_i)$ for every~$i$.  If $Q$ is
$\varepsilon$-far, then for every fixed~$S$,
\begin{equation}
\Pr[\mathsf{Miss}\mid S]
\leq(1-\varepsilon)^m
\leq e^{-\varepsilon m}\leq\frac16.
\label{eq:comparison-miss}
\end{equation}
Indeed, $H_S$ is a valid tree on~$X$, and the independent uniform
verification samples are drawn independently of~$S$.

%=============================================================
\subsection{The conditional sampling argument}
%=============================================================

Let~$Y_1,Y_2,\ldots,Y_m$ be the verification quartets sampled by the tester. 
Let~$\mathsf{Coll}$ be the event that some verification quartet~$Y_i$
contains two taxa of~$U\setminus S$ whose attachment edges
relative to~$S$ in~$T=\operatorname{\texttt{Learn}}(U)$ coincide.
Recall that~$T$ is always a phylogenetic tree, even when~$Q|_U$
is not tree-like.

The verification taxa also participate in learning~$T$, so we cannot
treat $T$ as independent of the verification samples.
The key idea is to condition on the sampled union~$U$, which fixes
$T$ for the given input~$Q$.
For a pair of verification taxa $x,y\notin S$, we additionally
condition on $W=S\cup\{x,y\}$.
By symmetry of the sampling, $\{x,y\}$ is then uniform over all
leaf pairs of the fixed tree $T|_W$.
Thus, we may view $S$ as the leaves remaining after a uniformly
chosen pair is removed.
The collision probability can now be bounded by counting
colliding pairs in this fixed tree, using
Lemma~\ref{lem:collision-count}.

\begin{lem}[Probability of a sampled collision]
\label{lem:collision-probability}
If $s<n$, then
\[
\Pr[\mathsf{Coll}]\leq\frac{24m}{s+1}.
\]
If $s=n$, then $\Pr[\mathsf{Coll}]=0$.
\end{lem}

\begin{proof}[Proof.]
Generate each $Y_i$ by first sampling a uniform ordered tuple
$(A_{i,1},A_{i,2},A_{i,3},A_{i,4})$ of distinct taxa and then forgetting its
order.  This gives precisely the specified distribution.
Fix $i$ and one pair of positions $1\leq a<b\leq4$, and set
$x=A_{i,a}$ and $y=A_{i,b}$.  Let~$\mathcal E$ be the event
$x,y\notin S$.  Conditional on~$E$, put
$W=S\cup\{x,y\}$, so $|W|=s+2$.

Fix sets~$D\subseteq A\subseteq X$ such that~$A$ and~$D$ are
possible realizations of~$U$ and~$W$, respectively, and
$\Pr[U=A,\ W=D\mid \mathcal{E}]>0$. Condition further on~$U=A$ and~$W=D$. 
The unordered pair $\{x,y\}$ is then uniform on~$\binom D2$.  
To verify this assertion, fix any two unordered pairs
$\mathbf{p},\mathbf{q}\in\binom D2$. Choose a permutation~$\pi$ of~$D$
mapping~$\mathbf{p}$ to~$\mathbf{q}$, and extend it to~$X$ by fixing every
taxon outside~$D$. Apply~$\pi$ to~$S$ and every verification
tuple. This gives a bijection between sampling outcomes
and preserves their probabilities, because~$S$ and the
verification tuples are sampled uniformly and independently.
This relabeling preserves~$\mathcal E$, since $x,y\notin S$ if and
only if $\pi(x),\pi(y)\notin\pi(S)$. Moreover, $\pi(D)=D$,
and $\pi(A)=A$ because $D\subseteq A$ and~$\pi$ fixes
every taxon outside~$D$. Hence the relabeling maps outcomes
satisfying~$\mathcal E$, $U=A$, $W=D$, and $\{x,y\}=\mathbf{p}$ bijectively
to equally probable outcomes satisfying the same conditioning
events and $\{x,y\}=\mathbf{q}$. Therefore, all pairs in~$\binom D2$
have equal conditional probability, proving uniformity.
The input~$Q$ remains fixed throughout; the relabeling
is applied only to the sampled sets and tuples.

The deterministic tree $\operatorname{\texttt{Learn}}(A)$ is fixed under
this conditioning.  Therefore
$R=\operatorname{\texttt{Learn}}(A)|_D$ is a fixed phylogenetic tree 
on~$s+2$ leaves. By transitivity of restriction, $x,y$ collide relative 
to~$S$ in~$T$ exactly when they are a colliding pair in~$R$ after
deleting them.  Lemma~\ref{lem:collision-count} bounds the
conditional probability by
\[
\frac{2(s+2)}{\binom{s+2}{2}}=\frac4{s+1}.
\]
Averaging over $A,D$ gives the same bound conditional on~$E$.
Multiplying by $\Pr[\mathcal E]\leq 1$ bounds the probability that this
pair of positions is an outside colliding pair.  A union bound
over the six pairs of positions and the~$m$ verification quartets
gives~$24m/(s+1)$.  When~$s=n$, there are no outside taxa and 
hence~$\Pr[\mathsf{Coll}]=0$.
\end{proof}

\iffalse
Lemma~\ref{lem:collision-probability} uses symmetry conditional
on $U$ and $W$.  Estimate~\eqref{eq:comparison-miss} uses
independence conditional on $S$.  Together these estimates control
the dependence introduced by using verification taxa during reconstruction.
\fi

%======================================================================
\subsection{Correctness, query complexity, and amplification}
%======================================================================

With the necessary lemmas established, we now prove
Theorem~\ref{thm:polynomial-tester}.

\begin{proof}[Proof of Theorem~\ref{thm:polynomial-tester}.]
Suppose first that $Q=Q_{T^\star}$ is tree-like.
Lemma~\ref{lem:small-set-learning} gives
$B=T^\star|_S$ and $T=T^\star|_U$.
Lemma~\ref{lem:insertion-search} makes every one-leaf attachment
correct.  Every restriction check and every verification query
therefore succeeds.  The tester always accepts in this case. 

Now let $Q$ be $\varepsilon$-far.  If the tester accepts, then
condition~\eqref{eq:comparison-tree} and the restriction checks imply
\[
T|_S=H_S|_S \; \text{ and }
\,\; T|_{S\cup\{x\}}=H_S|_{S\cup\{x\}}
\;\text{ for }\; x\in U\setminus S.
\]
If $\mathsf{Coll}$ also fails, the outside taxa of each $Y_i$
have pairwise distinct attachment edges.  Applying
Lemma~\ref{lem:distinct-attachments} with $W=Y_i\setminus S$
gives $Q_T(Y_i)=Q_{H_S}(Y_i)$.  Acceptance also entails
$Q(Y_i)=Q_T(Y_i)$, so $\mathsf{Miss}$ occurs.  
Thus, acceptance together with the absence of collisions implies
$\mathsf{Miss}$. Consequently,
\[
\{\text{accept}\}\subseteq\mathsf{Miss}\cup\mathsf{Coll}.
\]
By~\eqref{eq:comparison-miss} and Lemma~\ref{lem:collision-probability}, 
when $s<n$, we have $s=\min\{n,144m\}\geq 144m$ and then 
\[
\Pr[\text{accept}]
\leq\frac16+\frac{24m}{144m+1}<\frac13.
\]
If $s=n$, the collision probability is zero and the bound is
again at most~$1/6$.  This proves soundness for every~$n\geq4$.

Finally, put $k=|U|$ and
$D=\lceil\lg(296m)\rceil$.  Since $s\leq144m$ and
$k\leq s+4m\leq148m$, the two learning runs cost at most~$sD+kD$ queries.  
The~$k-s$ insertion checks cost at most~$(k-s)D$ queries, 
and verification on $Y_1,Y_2,\ldots,Y_m$ costs at most~$m$ queries.  The sum is
\[
2kD+m\leq296m\lceil\lg(296m)\rceil+m.
\]
\iffalse
Every executed tree operation concerns at most $148m$ leaves.
Searching candidate subtrees, taking restrictions, and comparing
labeled trees have elementary polynomial-time implementations.
Thus, apart from sampling and representing known taxon labels,
the computation is polynomial in~$1/\varepsilon$.
\fi
Every tree manipulated by the tester has at most~$148m$ leaves.
The learning runs and insertion checks make $O(m)$ calls to
$\operatorname{\texttt{Locate}}$, each with~$O(\lg m)$ iterations.
Searching candidate subtrees, taking restrictions, comparing
labeled trees, and extracting displayed quartets all admit
implementations taking time polynomial in~$m$.
Since $m=O(1/\varepsilon)$, the total running time is polynomial
in~$1/\varepsilon$, treating oracle queries and operations on
taxon labels as unit-cost and excluding the cost of sampling.
\end{proof}

\begin{cor}[Arbitrary confidence]
\label{cor:polynomial-amplification}
For $0<\varepsilon<1$ and $0<\delta\leq1/3$, there is a
polynomial-time adaptive one-sided tester using
\[
O\!\left(\varepsilon^{-1}\log(2/\varepsilon)\log(1/\delta)\right)
\]
queries, independently of $n$, with soundness error at most $\delta$.
\end{cor}

\begin{proof}[Proof.]
Run $R=\lceil\ln(1/\delta)/\ln3\rceil$ independent executions
of $\operatorname{\texttt{SampleTest}}$ and accept exactly when every
execution accepts.  Completeness remains~$1$.  For each fixed
$\varepsilon$-far input, acceptance has probability at most
$3^{-R}\leq\delta$.  Multiplying the query bound of
Theorem~\ref{thm:polynomial-tester} by~$R$ gives the result.
\end{proof}

\begin{cor}[A non-adaptive polynomial bound]
\label{cor:polynomial-nonadaptive}
Full quartet tree-likeness has an explicit polynomial-time
non-adaptive one-sided tester using
$O(\varepsilon^{-4}\log(1/\delta))$ queries for
$0<\varepsilon<1$ and $0<\delta\leq1/3$, independently of~$n$.
\end{cor}

\begin{proof}[Proof.]
\iffalse
Select $S,Y_1,Y_2,\ldots,Y_m$ and hence we have~$U$ before receiving any
answers.  Query every four-subset of~$U$ in advance and simulate
the complete execution of $\operatorname{\texttt{SampleTest}}$ from this
cache.  All its queries concern taxa in $U$, so the acceptance
decision is identical for every fixed choice of samples and oracle
answers.  One execution uses at most~$\binom{148m}{4}=O(\varepsilon^{-4})$ 
queries. Apply the independent repetition of
Corollary~\ref{cor:polynomial-amplification}, choosing the query
sets of all repetitions before reading any answers then the corollary follows.
\fi
Sample~$S,Y_1,Y_2,\ldots,Y_m$ before making any oracle queries,
and set $U=S\cup\bigcup_{i=1}^mY_i$.
Query every four-subset of~$U$ in a single nonadaptive batch,
then simulate the complete execution of
$\operatorname{\texttt{SampleTest}}$ using the cached answers.
Every query made by the tester concerns taxa in~$U$.
Thus, for every fixed input~$Q$ and choice of samples,
the simulation makes exactly the same acceptance decision
as the original tester. Since $|U|\leq148m$, one execution
uses at most $\binom{148m}{4}=O(\varepsilon^{-4})$ queries.

For amplification, apply the independent repetition scheme
of Corollary~\ref{cor:polynomial-amplification}.
Choose the samples for all repetitions in advance, and query
the union of their four-subset query sets before reading any
answers. Simulate each repetition using these cached answers.
The acceptance decisions are unchanged, so the amplification
guarantee applies, yielding the claimed nonadaptive query bound.
\end{proof}

%%%%%%%%%%%%%%%%%%%%%%%%%%%%%%%%%%%%%%%%%%%%%%%%%%%%%%%%%%%%%%%%%%%
\section{Explicit Reconstruct-and-Verify Testers}
\label{sec:learn-audit}
%%%%%%%%%%%%%%%%%%%%%%%%%%%%%%%%%%%%%%%%%%%%%%%%%%%%%%%%%%%%%%%%%%%

The algorithms in this section output a candidate tree on all
$n$ taxa whenever they accept.  This supports the reconstruct-or-reject
task formalized in Definition~\ref{def:reconstruct-or-reject}.
They first learn a global candidate and then verify it using fresh
uniform quartet samples.

%==================================================================
\subsection{Anchored triplets and exact learning}
%==================================================================

Fix an anchor taxon $r\in X$.  For distinct
$a,b,c\in X\setminus\{r\}$, one quartet query defines a rooted-triplet
answer on $\{a,b,c\}$ by
\[
Q(\{a,b,c,r\})=ab\mid cr \Rightarrow ab\mid c, \;\; 
Q(\{a,b,c,r\})=ac\mid br \Rightarrow ac\mid b, \;\; 
Q(\{a,b,c,r\})=bc\mid ar \Rightarrow bc\mid a.
\]
\begin{figure}[htbp]
\centering
\begin{tikzpicture}[
    line cap=round,
    line join=round,
    branch/.style={line width=0.85pt},
    leaf/.style={circle,fill=black,inner sep=1.45pt},
    inner/.style={circle,fill=black,inner sep=1.65pt},
    anchorleaf/.style={leaf,fill=red!75!black},
    marked/.style={
        circle,
        draw=blue!70!black,
        fill=white,
        line width=1pt,
        inner sep=2.05pt
    },
    taxon/.style={font=\large\itshape},
    note/.style={font=\footnotesize,align=center}
]

%------------------------------------------------
% (i) Anchored unrooted quartet ab|cr
%------------------------------------------------
\begin{scope}
    \coordinate (u) at (0.8,0);
    \coordinate (v) at (2.2,0);
    \coordinate (a) at (0,0.72);
    \coordinate (b) at (0,-0.72);
    \coordinate (c) at (3.0,0.72);
    \coordinate (r) at (3.0,-0.72);

    \draw[branch] (a)--(u)--(v)--(c);
    \draw[branch] (b)--(u);
    \draw[branch] (v)--(r);

    \node[leaf]       at (a) {};
    \node[leaf]       at (b) {};
    \node[leaf]       at (c) {};
    \node[anchorleaf] at (r) {};
    \node[inner]      at (u) {};
    \node[marked]     at (v) {};

    \node[taxon,above left=1pt and 3pt of a] {$a$};
    \node[taxon,below left=1pt and 3pt of b] {$b$};
    \node[taxon,above right=1pt and 3pt of c] {$c$};
    \node[taxon,text=red!75!black,
          below right=1pt and 3pt of r] {$r$};
    \node[taxon,text=blue!70!black,
          above left=1pt and 2pt of v] {$v$};

    \node[note] at (1.5,-1.68)
        {(i) anchored quartet $ab\mid cr$};
\end{scope}

% Transformation arrow
\draw[
    -{Stealth[length=2.8mm,width=2mm]},
    line width=0.9pt
]
(3.65,0) --
node[above=1mm,note] {\footnotesize delete $r$\\[-1.5pt]\footnotesize and root at $v$}
(5.15,0);

%------------------------------------------------
% (ii) Resulting rooted triple ab|c
%------------------------------------------------
\begin{scope}
    \coordinate (rho) at (6.45,0.75);
    \coordinate (w)   at (5.65,-0.10);
    \coordinate (aa)  at (5.05,-0.92);
    \coordinate (bb)  at (6.25,-0.92);
    \coordinate (cc)  at (7.35,-0.10);

    \draw[branch] (rho)--(w);
    \draw[branch] (rho)--(cc);
    \draw[branch] (w)--(aa);
    \draw[branch] (w)--(bb);

    \node[marked] at (rho) {};
    \node[inner]  at (w) {};
    \node[leaf]   at (aa) {};
    \node[leaf]   at (bb) {};
    \node[leaf]   at (cc) {};

    \node[note,text=blue!70!black,
          above=5pt of rho] {root $v$};
    \node[taxon,below left=1pt and 3pt of aa] {$a$};
    \node[taxon,below right=1pt and 3pt of bb] {$b$};
    \node[taxon,right=4pt of cc] {$c$};

    \node[note] at (6.2,-1.68)
        {(ii) rooted triple $ab\mid c$};
\end{scope}

\end{tikzpicture}

\caption{Anchoring a quartet. For the fixed anchor $r$, deleting
$r$ from the quartet $ab\mid cr$ and rooting the remaining tree at
the former neighbor $v$ of $r$ produces the rooted triple
$ab\mid c$. The filled red vertex denotes the anchor, and the
outlined blue vertex denotes $v$ before and after the
transformation. The other two rules follow by permuting $a,b,c$.}
\label{fig:anchored-quartet-triplet}
\end{figure}
If $Q$ is induced by an unrooted tree $T$, delete the leaf $r$
and root the remaining tree at the former neighbor of~$r$ (see
Fig.~\ref{fig:anchored-quartet-triplet}).  The three rules above return
exactly the rooted triplets displayed by this rooted binary tree.

\begin{lem}[Anchored exact learning]
\label{lem:anchored-learning}
Let $Q$ be a tree-like full quartet system on an $n$-taxon set.
For any fixed $r\in X$, the unique tree inducing $Q$ can be
reconstructed by a deterministic algorithm using at most
\[
L(n):=\left\lceil (n-1)\lg(n-1)\right\rceil
\]
adaptive quartet queries, all of which contain~$r$.
\end{lem}

\begin{proof}
Emamjomeh-Zadeh and Kempe~\cite{EmamjomehZadehKempe18} give a
deterministic algorithm that reconstructs a rooted binary tree on $N$
leaves from at most $N\lg N$ adaptive ordinal (rooted-triplet)
queries.  Apply their algorithm with $N=n-1$ and simulate every
triplet query by the anchored quartet query above.  The learned rooted
tree is precisely the tree obtained from $T$ by deleting~$r$.
Attaching $r$ to its root recovers~$T$.
\end{proof}

\paragraph{Remark} The learned rooted tree is precisely the phylogenetic tree obtained 
from~$T$ by deleting the anchor leaf~$r$, retaining its former neighbor as the root. 
Since this root records the original attachment point of~$r$, attaching a new leaf 
labeled~$r$ directly to the root and then forgetting the root orientation 
recovers~$T$ up to leaf-label-preserving isomorphism.

The exact-learning guarantee above applies when the queried
quartets are induced by a phylogenetic tree.  The tester, however, must also be
well defined on an arbitrary full quartet system, whose answers may
be mutually inconsistent.  We therefore run the exact learner for at
most~$L(n)$ queries.  If the observed query--answer transcript causes
the learner to reach a case not covered by its rules, or if the
learner has not returned a phylogenetic tree within this budget, we
use an arbitrary rooted binary tree on $X\setminus\{r\}$ as the
candidate.  This modification has no effect when~$Q$ is tree-like,
because Lemma~\ref{lem:anchored-learning} then guarantees that the
learner returns the correct rooted tree within~$L(n)$ queries.  Thus,
the learning phase always produces a candidate within the stated
budget. Correctness of that candidate is needed only in the
tree-like case.

%==================================================================
\subsection{The adaptive tester}
%==================================================================

\paragraph{Adaptive reconstruct-and-verify tester.}
\label{alg:learn-audit}
\begin{enumerate}[leftmargin=*,label=\arabic*.]
\item Fix $r\in X$ and run the anchored exact learner using at
most $L(n)$ quartet queries.
\item Attach $r$ to the root of the returned rooted tree to obtain
an unrooted candidate $\widehat T$. 
\item Set
\[
m=m_\star(\varepsilon,\delta):=
\left\lceil
\frac{\ln(1/\delta)}{\ln(1/(1-\varepsilon))}
\right\rceil.
\]
Independently sample $Y_1,Y_2,\ldots,Y_m$ uniformly from $\binom X4$,
with replacement.
\item Query the sampled quartets.  Return \textsf{reject} if
$Q(Y_i)\ne Q_{\widehat T}(Y_i)$ for any $i$. Otherwise, return
\textsf{accept}.
\end{enumerate}

\begin{thm}[Adaptive reconstruct-and-verify tester]
\label{thm:learn-audit}
The reconstruct-and-verify procedure above is an adaptive one-sided
tester for full quartet tree-likeness.  It runs in polynomial time and
makes at most
\[
\left\lceil (n-1)\lg(n-1)\right\rceil+
\left\lceil
\frac{\ln(1/\delta)}{\ln(1/(1-\varepsilon))}
\right\rceil
\]
queries.  In particular, its query complexity is
$O(n\log n+\varepsilon^{-1}\log(1/\delta))$.
\end{thm}

\begin{proof}[Proof.]
Suppose first that $Q=Q_T$ for a phylogenetic $X$-tree $T$.
Lemma~\ref{lem:anchored-learning} implies
$\widehat T=T$.  Every verification query therefore agrees with
$Q_{\widehat T}$, so the algorithm accepts with probability~$1$.

Now suppose that $Q$ is $\varepsilon$-far from tree-likeness.
Condition on the entire learning transcript and hence on the
particular candidate $\widehat T$ that it produces.  By the definition
of~$\varepsilon$-farness,
\[
\Pr_{Y\sim\operatorname{Unif}(\binom{X}{4})}
   [Q(Y)\ne Q_{\widehat T}(Y)]\geq\varepsilon.
\]
The verification samples are fresh and independent of the
learning phase.  Consequently, the probability that all~$m$ samples
miss the disagreement set is at most $(1-\varepsilon)^m\leq \delta$. 
This bound holds for every conditioned transcript and therefore
also unconditionally.  The displayed query bound follows by adding
the reconstruction and verification costs.  Both phases are implementable in
polynomial time.
\end{proof}

%==================================================================
\subsection{A cached-anchor non-adaptive variant}
%==================================================================

\paragraph{The rooted compatibility primitive.}
For three leaves $a,b,c$ of a rooted tree, $ab\mid c$ means that
the least common ancestor of $a,b$ is a proper descendant of the
least common ancestor of $a,b,c$.  A collection of such rooted
triplets is compatible if one rooted tree displays all of them.
The BUILD algorithm of Aho et al.~\cite{AhoSagivSzymanskiUllman81}
decides compatibility and returns a displaying rooted tree when one
exists, in polynomial time.  The complete triplet system used here
determines a unique rooted binary tree when compatible.

For completeness, this guarantee has a direct recursive implementation.
On a leaf set $A$, form a graph with edge $ab$ whenever the input
contains $ab\mid c$ for some $c\in A\setminus\{a,b\}$.
A singleton returns a leaf.  Otherwise, reject unless the graph has
exactly two connected components, recurse on them, and join their
roots under a new root.  Verify the returned tree against every
input triplet and reject on any disagreement.  A compatible complete
system passes this procedure.  If the true root separates $A_1,A_2$,
the graph has no edges between them, and every pair within $A_i$
is an edge by choosing $c$ in the other part.  Its two components
are therefore exactly $A_1,A_2$, and induction proves reconstruction
and uniqueness.  The final verification ensures soundness on every
incompatible input, using the already cached answers.

\begin{cor}
\label{cor:cached-anchor}
For $n\geq 5$, there is a polynomial-time non-adaptive one-sided
tester using at most
\[
\binom{n-1}{3}+
\left\lceil
\frac{\ln(1/\delta)}{\ln(1/(1-\varepsilon))}
\right\rceil
\]
queries.
\end{cor}

\begin{proof}[Proof.]
Fix $r\in X$ and, before receiving any answers, select all
$\binom{n-1}{3}$ quartets containing $r$ together with $m$ uniformly
random four-sets from $\binom{X\setminus\{r\}}4$, where
$m=m_\star(\varepsilon,\delta)$.  The anchored
answers form a complete rooted-triplet system on
$X\setminus\{r\}$.  Apply the rooted compatibility primitive above and reject if
it reports incompatibility. Otherwise, reconstruct its unique rooted binary tree,
attach $r$ to the root, and call the resulting tree $\widehat T$.

If $Q$ is tree-like, the rooted-triplet system is compatible and
the reconstructed tree is the inducing tree, so the tester never
rejects.  Otherwise, whenever the reconstruction phase does not
reject, $\widehat T$ agrees with every anchored quartet.  If $Q$ is
$\varepsilon$-far from the tree-likeness, at least $\varepsilon\binom n4$ unanchored
quartets disagree with $\widehat T$. Hence, a uniform unanchored verification query
finds such a disagreement with probability at least~$\varepsilon$.
The same calculation as in Theorem~\ref{thm:learn-audit} gives failure
probability at most~$\delta$.  All query locations were selected in
advance, and hence the tester is non-adaptive.
\end{proof}

\paragraph{The verification query as a quintet test.}
Suppose compatibility succeeds and produces $\widehat T$.
For an unanchored four-set $Y$, let $Z=Y\cup\{r\}$.
The four anchored entries of $Q|_Z$ all agree with
$Q_{\widehat T}|_Z$.  If the fifth entry $Q(Y)$ agrees as well,
then $Q|_Z$ is the tree-induced quintet of $\widehat T|_Z$.
If $Q(Y)$ disagrees, the two quartet systems differ in exactly one entry.
By Fact~\ref{fact:quintet-distance}, that modified word cannot be
another resolved quintet.  Thus, after anchored compatibility succeeds,
checking agreement on $Y$ is equivalent to checking resolution of~$Z$.
This equivalence connects our verification step to the tests
of quintets containing a fixed taxon used
by Chang, Lin, and Rossmanith~\cite{CLR11}.
Our algorithm queries all~$\binom{n-1}{3}$ anchored entries once
and reuses them, so each sampled quintet requires at most one
additional quartet query. This gives the additive bound
of Corollary~\ref{cor:cached-anchor}.

%%%%%%%%%%%%%%%%%%%%%%%%%%%%%%%%%%%%%%%%%%%%%%%%%%%%%%%%%%%%%%%%%%%%%%%%%%
\section{A Hereditary Description of Quartet Tree-Likeness}
\label{sec:reduction}
%%%%%%%%%%%%%%%%%%%%%%%%%%%%%%%%%%%%%%%%%%%%%%%%%%%%%%%%%%%%%%%%%%%%%%%%%%

This section explains how quartet tree-likeness fits the general
theory of hereditary testing.  Encoding topologies as colors of
ordered tuples makes the property invariant under relabeling.
Closure under restriction then permits the application of the
Austin--Tao theorem.  Theorem~\ref{thm:polynomial-tester} gives the
explicit quantitative bound for the present property.

%==================================================================
\subsection{Colored structures and the testing theorem}
%==================================================================

A finite palette $\mathbf K=(K_j)_{j\geq0}$ consists of finite
nonempty color sets, all but finitely many of which are singletons.
We use $K_4=\{0,1,2\}$ and singleton sets at every other arity.
A $\mathbf K$-colored structure on $V$ assigns a color in $K_j$
to every ordered $j$-tuple of distinct vertices.  The singleton
components carry no information, so such a structure is specified
by a function $c:V_{\ne}^4\to\{0,1,2\}$, where $V_{\neq}^4$ denotes 
the set of ordered quadruples of pairwise distinct vertices from~$V$.

For $H=(V,c)$ and $Z\subseteq V$, the induced structure $H[Z]$
restricts $c$ to $Z_{\ne}^4$.  An isomorphism from $(V,c)$ to
$(V',c')$ is a bijection $\pi:V\to V'$ satisfying
$c'(\pi x_1,\ldots,\pi x_4)=c(x_1,\ldots,x_4)$ on every tuple.
A hereditary $\mathbf K$-property assigns a family $\mathcal P(V)$
to each finite $V$ and is closed under these isomorphisms and
induced restrictions.

The relevant distance counts four-sets on which the complete
induced color data differ,
\[
d_4(H,H')=
\frac{|\{Y\in\binom V4:H[Y]\ne H'[Y]\}|}{\binom{|V|}{4}},
\qquad d_4(H,\mathcal P)=\min_{H'\in\mathcal P(V)}d_4(H,H').
\]
This is the distance in the following specialization of hereditary
directed-hypergraph testing.

\begin{thm}[Austin--Tao hereditary testing]
\label{thm:austin-tao}
For the finite palette above and every hereditary
$\mathbf K$-property $\mathcal P$, the following holds.
For every $\tau>0$, there are $t=t(\tau)\geq5$ and
$\rho=\rho(\tau)\in(0,1)$ such that, if $|V|\geq t$ and
$d_4(H,\mathcal P)\geq\tau$, then a uniform $t$-subset $Z$ of $V$
satisfies $H[Z]\notin\mathcal P(Z)$ with probability at least
$\rho$.
\end{thm}

\iffalse
This is the contrapositive form of the testability conclusion in
Austin and Tao~\cite{AustinTao10}, which acts as being globally far 
forces a positive probability of finding a local violation. 
Applying their repair-distance parameter with $\tau/2$ gives the
stated conclusion for inputs at distance at least $\tau$.
\fi
Apply the testability theorem of Austin and Tao~\cite{AustinTao10} 
with distance parameter~$\tau/2$.
Their theorem gives~$t$ and~$\rho$ such that, if a uniformly
random $t$-subset $Z$ satisfies $H[Z]\in\mathcal P(Z)$ with
probability at least~$1-\rho$, then
$d_4(H,\mathcal P)\leq \tau/2$.
Consequently, if $d_4(H,\mathcal P)\geq \tau$, the probability
that $H[Z]\notin\mathcal P(Z)$ must exceed~$\rho$. 
Roughly speaking, if a structure is globally far from
satisfying~$\mathcal P$, then a random small restriction
violates~$\mathcal P$ with probability bounded away from zero.

%==================================================================
\subsection{A relabeling-equivariant quartet encoding}
%==================================================================

Let $X_{\neq}^4$ denote the ordered four-tuples of distinct taxa.
For a full quartet system $Q$, define
\[
c_Q(x_1,x_2,x_3,x_4)=
\begin{cases}
0,&Q(\{x_1,x_2,x_3,x_4\})=x_1x_2\mid x_3x_4,\\
1,&Q(\{x_1,x_2,x_3,x_4\})=x_1x_3\mid x_2x_4,\\
2,&Q(\{x_1,x_2,x_3,x_4\})=x_1x_4\mid x_2x_3.
\end{cases}
\]
Write~$H_Q=(X,c_Q)$. Each four-set has~$4!=24$ orderings,
and~$c_Q$ assigns one of the three colors~$0,1,2$ to each
ordering. A single quartet answer determines all~$24$
color assignments on that four-set.  More explicitly, the consistency
condition on those $24$ colors requires that, after decoding a
color by the displayed rule in its own tuple order, every ordering
gives the same unordered bipartition.  These are constraints on four
vertices.

For any bijection $\pi:X\to X'$, define the relabeled system by
$(\pi Q)(\pi Y)=\pi(Q(Y))$.  The encoding satisfies
\[
c_{\pi Q}(\pi x_1, \pi x_2, \pi x_3, \pi x_4)=c_Q(x_1,x_2,x_3,x_4),
\text{ and }\, H_{Q|_S}=H_Q[S].
\]
The three colors refer to positions within a tuple, so these
identities hold under every relabeling of taxon names.

\begin{lem}[Exact distance preservation]
\label{lem:distance-preservation}
For full quartet systems $Q,Q'$ on $X$,
\[
d_4(H_Q,H_{Q'})
=\frac{|\{\mathbf x\in X_{\ne}^4:c_Q(\mathbf x)\ne c_{Q'}(\mathbf x)\}|}
{|X_{\ne}^4|}
=\operatorname{dist}(Q,Q').
\]
\end{lem}
\begin{proof}
Fix a four-set $Y$. For each ordering $\mathbf{x}$ of~$Y$,
the encoding assigns distinct colors to the three possible
quartet topologies. Hence $c_Q(\mathbf{x})=c_{Q'}(\mathbf{x})$
if and only if $Q(Y)=Q'(Y)$.
Thus the two encodings agree on every ordering of $Y$ if the
topologies agree, and disagree on every ordering otherwise.
Let $n=|X|$ and
\[
D:=\left\{Y\in\binom X4:Q(Y)\neq Q'(Y)\right\}.
\]
The induced structures $H_Q[Y]$ and $H_{Q'}[Y]$ differ exactly
when $Y\in D$, so
\[
d_4(H_Q,H_{Q'})=\frac{|D|}{\binom n4}
=\operatorname{dist}(Q,Q').
\]
Each four-set has exactly $4!=24$ orderings, and every injective
ordered four-tuple belongs to exactly one four-set. Therefore,
\[
\frac{|\{\mathbf{x}\in X_{\ne}^4:
c_Q(\mathbf{x})\ne c_{Q'}(\mathbf{x})\}|}{|X_{\ne}^4|}
=\frac{24|D|}{24\binom n4}
=\operatorname{dist}(Q,Q').
\]
This proves both equalities.
\end{proof}

%==================================================================
\subsection{The hereditary property and its local recognition}
%==================================================================

Fix a labeled five-set.  Its $3^5=243$ full quartet systems give
a family $\mathcal E_5$ of encoded structures.  Let
$\mathcal R_5\subseteq\mathcal E_5$ consist of the $15$ structures
induced by labeled quintet trees, and set
$\mathcal F_5=\mathcal E_5\setminus\mathcal R_5$.
Thus $\mathcal F_5$ has~$228$ forbidden labeled patterns.
By Corollary~\ref{cor:tree-like-quintets} for $|X|\geq 5$,
and trivially for~$|X|=4$, 
\[
Q\text{ is tree-like}
\iff
H_Q[Z]\cong R\text{ for some }R\in\mathcal R_5
\text{ for every }Z\in\binom X5, 
\]
where $\cong$ denotes isomorphism of $\mathbf K$-colored
structures, meaning that a bijection between their vertex
sets preserves the color of every ordered tuple. 
Each isomorphism here transports the labels of $Z$ to the fixed
reference five-set.  The family $\mathcal R_5$ includes every
labeling of each possible quintet topology.

For each finite $V$ of size at least four, let
$\mathcal P_{\rm enc}(V)=\{H_Q:Q\in\mathcal P_{\rm tree}(V)\}$.
For smaller sets include the unique empty four-ary structure.
This property is recognized by the four-vertex decoding constraints
and the exclusion of the five-vertex patterns $\mathcal F_5$.
The encoding identities and Theorem~\ref{thm:treelike-hereditary}
show that it is a hereditary $\mathbf K$-property.
Lemma~\ref{lem:distance-preservation} also gives
\[
d_4(H_Q,\mathcal P_{\rm enc})
=\operatorname{dist}(Q,\mathcal P_{\rm tree}).
\]

\begin{prop}[Qualitative hereditary tester]
\label{prop:exist}
For every $0<\varepsilon,\delta<1$, full quartet tree-likeness
has a non-adaptive one-sided tester using
$q_{\rm her}(\varepsilon,\delta)$ queries for a bound independent
of $n$.
\end{prop}
\begin{proof}
Apply Theorem~\ref{thm:austin-tao} with
$\mathcal P=\mathcal P_{\rm enc}$ and $\tau=\varepsilon$.
For the resulting $t,\rho$, sample a uniform $t$-set $Z$, query
all $\binom t4$ quartet entries, and test its encoded restriction
for membership.  The local conditions above decide membership by
checking every five-set.  If the input is $\varepsilon$-far,
the distance identity and the theorem give rejection probability
at least $\rho$.  Every tree-like input passes every restriction.
Use
\[
R=\left\lceil\frac{\ln(1/\delta)}{-\ln(1-\rho)}\right\rceil
\]
independent sampled restrictions and reject if any fails.
For $n\geq t$ the query bound is $R\binom t4$.
For $n<t$, query the complete input and test it exactly using the
same local characterization, with the four-taxon case accepted.
This uses at most $\binom t4$ queries.  All sampled query sets
are chosen before reading any answers.
\end{proof}

\paragraph{Quantitative interpretation}
The cited hereditary theorem gives existence of a sample size for
each distance parameter~$\varepsilon$.  Its general regularity and removal
arguments supply no polynomial bound specialized to quartet
tree-likeness, and their quantitative estimates can involve
iterated-exponential growth.  We leave $q_{\rm her}$ unspecified.
Theorem~\ref{thm:polynomial-tester} and
Corollary~\ref{cor:polynomial-nonadaptive} have already provided the 
explicit polynomial bounds for this property.

\paragraph{Remark} The dense incompatibility constructions of
Alon--Snir--Yuster~\cite{AlonSnirYuster14} concern partial quartet
sets sampled from their specified constraints.  Our input supplies
one topology on every four-set, and a sampled taxon subset exposes
its complete induced system.  These assumptions explain
the applicability of hereditary testing here.

%%%%%%%%%%%%%%%%%%%%%%%%%%%%%%%%%%%%%%%%%%%%%%%%%%%%%%%%%%%%%%
\section{Lower Bounds and the Role of Adaptivity}
\label{sec:LB}
%%%%%%%%%%%%%%%%%%%%%%%%%%%%%%%%%%%%%%%%%%%%%%%%%%%%%%%%%%%%%%

Ordinary testing requires an accept/reject decision for tree-like
and $\varepsilon$-far systems.  We first bound the queries of a general
two-sided tester and then state the strengthening from one-sided
completeness.  We subsequently analyze the extra query cost of exact
reconstruction on tree-like inputs.  The explicit upper bound in
Theorem~\ref{thm:polynomial-tester} and the ordinary lower bound below
leave a factor $O(\log(2/\varepsilon))$ for constant error probability.

%==================================================================
\subsection{A lower bound for property testing}
%==================================================================

The lower bound is obtained by randomly perturbing a fixed
tree-like quartet system.  Each quartet is changed independently with
probability $p=(1+\gamma)\varepsilon$, where $\gamma>0$ provides the
slack needed to make the perturbed system $\varepsilon$-far with high
probability.  On the other hand, a tester making at most $q$ queries
sees exactly the same transcript as on the original tree-like system
with probability at least $(1-p)^q$.  Taking $n\to\infty$ and then
$\gamma\downarrow0$ yields the leading constant.

\begin{thm}[Finite-$n$ lower bound for general testers]
\label{thm:lower-bound}
Fix $\varepsilon\in(0,2/3)$, $\delta\in(0,1/2)$, and $\gamma>0$
such that $p=(1+\gamma)\varepsilon\leq2/3$.  Let $\mathcal A$
be a randomized, possibly adaptive tester using at most $q$ queries,
accepting every tree-like input with probability at least $1-\delta$
and every $\varepsilon$-far input with probability at most $\delta$.
Set $N=\binom n4$ and
\[
\eta_n(\gamma)=(2n-5)!!
\exp\!\left(-\frac{\gamma^2\varepsilon}{2(1+\gamma)}N\right).
\]
Then $(1-\delta)(1-p)^q\leq\delta+\eta_n(\gamma)$.
Whenever $\eta_n(\gamma)<1-2\delta$, this implies
\[
q\geq
\frac{\ln((1-\delta)/(\delta+\eta_n(\gamma)))}
{-\ln(1-(1+\gamma)\varepsilon)}.
\]
\end{thm}

\begin{proof}
Fix a tree-like full quartet system $Q^\star$.
Independently for every $Y\in\binom X4$, define a random full quartet
system $Q$ by retaining $Q^\star(Y)$ with probability $1-p$ and,
with probability $p$, choosing uniformly between the other two
topologies.
Fix a phylogenetic tree $T$ and write
\[
D_T:=\left|\Bigl\{Y\in\binom X4:Q(Y)\ne Q_T(Y)\Bigr\}\right|.
\]
If $Q_T(Y)=Q^\star(Y)$, the disagreement probability is
$p$.  Otherwise, $Q(Y)$ equals $Q_T(Y)$ with probability $p/2$, so
the disagreement probability is $1-p/2\ge p$, where the final
inequality uses $p\le2/3$.  Thus
$\mu_T:=\mathbb E[D_T]\ge pN$.  With
$\theta=\gamma/(1+\gamma)$, a multiplicative Chernoff bound for the
independent indicator random variables (see~Theorem 4.5 in~\cite{mitzenmacher2017probability}) gives
\[
\begin{aligned}
\Pr[D_T<\varepsilon N]
&\le
\Pr[D_T<(1-\theta)\mu_T]  \\
&\le
\exp\!\left(-\frac{\theta^2\mu_T}{2}\right)
\le
\exp\!\left(
-\frac{\gamma^2\varepsilon}{2(1+\gamma)}N
\right).
\end{aligned}
\]
Every tree-like full quartet system is induced by a labeled
unrooted binary tree, and there are $(2n-5)!!$ such trees 
(see~\cite[Sect.~2.1]{SempleSteel2003}).  A union
bound therefore gives 
\begin{equation}
\Pr[\operatorname{dist}(Q,\mathcal{P}_{\text{tree}})<\varepsilon]
\leq \eta_n(\gamma).\label{eq:epsilon-far}
\end{equation}

We next couple executions of $\mathcal A$ on $Q^\star$ and
on $Q$ by using the same internal random coins~$R$.  Fix a value of~$R$ 
for which the clean execution on~$Q^\star$ accepts, and let~$S_R$ 
be the set of distinct quartet coordinates queried during that
execution.  Since $\mathcal{A}$ makes at most~$q$ queries, 
we have~$|S_R|\leq q$.  If no coordinate in~$S_R$ is corrupted, that is, 
if the quartet topology assigned to each four-set in~$S_R$ is the same 
in~$Q^*$ and~$Q$, then adaptivity causes no difficulty: the two executions
receive the same answer to every query, make the same subsequent
queries, and finally produce the same accepting output.  The
probability of this event is
\[
(1-p)^{|S_R|}\geq (1-p)^q.
\]
The clean execution accepts for a set of random coins of
probability at least $1-\delta$.  Averaging the last bound over those
coins yields 
\[
\Pr_{Q,R}[\mathcal A(Q)\text{ accepts and the transcripts coincide}]
\geq (1-\delta)(1-p)^q.
\]
Let $F$ be the event that~$Q$ is $\varepsilon$-far.  Since
$\Pr[F^c]\leq \eta_n(\gamma)$ (by Inequality~\eqref{eq:epsilon-far}), intersecting the preceding event 
with~$F$ loses probability at most~$\eta_n(\gamma)$. Indeed, 
let $\mathcal G$ denote the event that $\mathcal A(Q)$ accepts and
the two transcripts coincide.  Since $\mathcal G\cap F^c\subseteq F^c$ and
$\Pr[F^c]\leq\eta_n(\gamma)$,
\[
\begin{aligned}
\Pr_{Q,R}[\mathcal A(Q)\text{ accepts}\wedge F]
&\geq \Pr[\mathcal G\cap F]\\
&\geq \Pr[\mathcal G]-\Pr[F^c]\\
&\geq (1-\delta)(1-p)^q-\eta_n(\gamma).
\end{aligned}
\]
For every fixed full quartet system $\hat{Q}$, 
we have~$\Pr_R[\mathcal{A}(\hat{Q})\text{ accepts}]\leq \delta$ 
whenever~$\hat{Q}$ is~$\varepsilon$-far (by soundness).  Therefore
\[ 
\Pr_{Q,R}[\mathcal{A}(Q)\text{ accepts }\wedge F] = 
\mathbb E_Q\!\left[\mathbf 1_F \Pr_R[\mathcal{A}(Q)\text{ accepts}]\right] \leq \delta\Pr[F] \leq \delta.
\]
Combining the upper and lower bounds gives
\[
\delta\geq (1-\delta)(1-p)^q-\eta_n(\gamma),
\]
which proves the claimed inequality.  Taking logarithms gives
the stated query lower bound.
\end{proof}

\begin{cor}[Strengthening under one-sided completeness]
\label{cor:finite-one-sided}
Under the parameters and notation of Theorem~\ref{thm:lower-bound},
if $\mathcal A$ accepts every tree-like input with probability $1$,
then $(1-p)^q\leq\delta+\eta_n(\gamma)$.  If
$\eta_n(\gamma)<1-\delta$, then
\[
q\geq\frac{\ln(1/(\delta+\eta_n(\gamma)))}
{-\ln(1-(1+\gamma)\varepsilon)}.
\]
\end{cor}
\begin{proof}
Every choice of coins, apart from a probability-zero set, accepts
the clean input.  The same coupling therefore bounds the probability
of identical accepting transcripts below by $(1-p)^q$.
Intersecting with the far-input event and applying soundness gives
$\delta\geq(1-p)^q-\eta_n(\gamma)$.  Taking logarithms proves
the query bound.
\end{proof}

\begin{cor}[An asymptotic form of the lower bound]
\label{cor:sharp-testing-lb}
For fixed $\varepsilon\in(0,2/3)$ and
$\delta\in(0,1/2)$, every family of randomized, possibly adaptive
two-sided testers satisfies 
\[
\liminf_{n\to\infty}q(n,\varepsilon,\delta)
\geq
\frac{\ln((1-\delta)/\delta)}
     {\ln(1/(1-\varepsilon))}.
\]
If the tester accepts every tree-like input with probability~$1$, then  
\[
\liminf_{n\to\infty}q(n,\varepsilon,\delta)
\geq
\frac{\ln(1/\delta)}
     {\ln(1/(1-\varepsilon))}.
\]
Hence, as $\varepsilon,\delta\rightarrow 0$ after
$n\to\infty$, the lower bound is
\[
(1-o(1))\,\frac{\ln(1/\delta)}{\varepsilon}.
\]
\end{cor}

\begin{proof}
Fix an admissible $\gamma>0$.  Since
\[
\ln \eta_n(\gamma)
=
\ln((2n-5)!!)
-
\frac{\gamma^2\varepsilon}{2(1+\gamma)}N,
\]
and $\ln((2n-5)!!)=O(n\ln n)$ whereas
$N=\Theta(n^4)$, 
$\eta_n(\gamma)\to 0$ as
$n\to\infty$.  In particular, since $(1-\delta)(1-p)^q\leq \delta+\eta_n(\gamma)$, we have 
\[
(1-p)^q
\leq
\frac{\delta+\eta_n(\gamma)}{1-\delta}. 
\]
Since $\delta<1/2$ and $\eta_n(\gamma)\to0$, we have
$\eta_n(\gamma)<1-2\delta$ for all sufficiently large~$n$.
Theorem~\ref{thm:lower-bound} therefore gives
\[
\liminf_{n\to\infty}q(n,\varepsilon,\delta)
\geq
\frac{\ln((1-\delta)/\delta)}
     {-\ln(1-(1+\gamma)\varepsilon)}.
\]
Since this holds for every admissible $\gamma>0$, letting
$\gamma\to 0$ gives
\[
\liminf_{n\to\infty}q(n,\varepsilon,\delta)
\geq
\frac{\ln((1-\delta)/\delta)}
     {\ln(1/(1-\varepsilon))}.
\]
Applying Corollary~\ref{cor:finite-one-sided} in the same limiting argument gives 
\[
\liminf_{n\to\infty}q(n,\varepsilon,\delta)
\geq
\frac{\ln(1/\delta)}
     {\ln(1/(1-\varepsilon))}.
\]
Finally, as $\delta\to0$,
\[
\frac{\ln((1-\delta)/\delta)}{\ln(1/\delta)}
=
1+\frac{\ln(1-\delta)}{\ln(1/\delta)}
\longrightarrow 1,
\]
since $\ln(1-\delta)=-\delta+O(\delta^2)$.  Moreover, as
$\varepsilon\to0$,
\[
\frac{\varepsilon}{-\ln(1-\varepsilon)}
=
\frac{1}{1+\varepsilon/2+O(\varepsilon^2)}
\longrightarrow 1.
\]
Consequently,
\[
\begin{aligned}
\frac{\ln((1-\delta)/\delta)}
     {-\ln(1-\varepsilon)}
&=
\frac{\ln((1-\delta)/\delta)}
     {\ln(1/\delta)}
\cdot
\frac{\varepsilon}
     {-\ln(1-\varepsilon)}
\cdot
\frac{\ln(1/\delta)}{\varepsilon}\\
&=
(1-o(1))\,
\frac{\ln(1/\delta)}{\varepsilon},
\end{aligned}
\]
as $(\varepsilon,\delta)\to(0,0)$.  Here, $o(1)$ tends to zero
under every joint approach $(\varepsilon,\delta)\to(0,0)$.
The same asymptotic form follows for the one-sided-error lower bound,
with the first factor omitted.
\end{proof}

\paragraph{Remark.} 
For fixed $0<\varepsilon,\delta\leq1/4$, the lower bounds are
$\Omega(\varepsilon^{-1}\log(1/\delta))$ for sufficiently large $n$.
Indeed, $\varepsilon\leq-\ln(1-\varepsilon)\leq \int_0^\varepsilon\frac{1}{1-t}\,dt
\leq \frac{\varepsilon}{1-\varepsilon}\leq 4\varepsilon/3$,
and $\ln((1-\delta)/\delta)$ is bounded below by a positive constant
times $\ln(1/\delta)$ on this interval.
The one-sided expression matches the fresh verification sample count
in Theorem~\ref{thm:learn-audit}, up to rounding and the asymptotic
limit in $n$.  For ordinary testing, comparison with
Corollary~\ref{cor:polynomial-amplification} leaves the factor
$O(\log(2/\varepsilon))$.

%==================================================================
\subsection{The stronger reconstruct-or-reject task}
%==================================================================

\begin{defn}[Reconstruct-or-reject algorithm]
\label{def:reconstruct-or-reject}
A randomized quartet-query algorithm is a
\emph{$(\varepsilon,\delta)$ reconstruct-or-reject algorithm} if it
outputs either a labeled phylogenetic tree or \textsf{reject} and
satisfies the following conditions:
\begin{itemize}
\item if $Q=Q_T$ is tree-like, it outputs the unique inducing
tree $T$ with probability at least $1-\delta$; and
\item if
$\operatorname{dist}(Q,\mathcal P_{\rm tree})\geq\varepsilon$, it
outputs \textsf{reject} with probability at least $1-\delta$. 
\end{itemize}
It is \emph{one-sided} if it outputs $T$ with probability~$1$
on every tree-like input.  No requirement is imposed on inputs that
are neither tree-like nor $\varepsilon$-far.
\end{defn}

The algorithms of Theorem~\ref{thm:learn-audit} and
Corollary~\ref{cor:cached-anchor} satisfy this stronger definition
without further queries. Clearly, whenever they would accept, they output the
candidate tree~$\widehat T$.

\begin{thm}[Reconstruction lower bounds]
\label{thm:reconstruction-lb}
Let $\delta\in[0,2/3)$ and assume that $n\geq 5$.  Every reconstruct-or-reject
algorithm on $n$ taxa that uses at most $q$ quartet queries satisfies
the following bounds.
\begin{enumerate}[leftmargin=*,label=(\alph*)]
\item If the algorithm is adaptive, then 
\[
q\geq
\log_3\!\left((1-\delta)(2n-5)!!\right)
=\Omega(n\log n).
\]
\item If the algorithm is non-adaptive, then 
\[
q\geq
\left(1-\frac{3\delta}{2}\right)
\frac{\binom n4}{4n-15}
=\Omega(n^3).
\]
\end{enumerate}
\end{thm}

\begin{proof}
There are $M_n=(2n-5)!!$ phylogenetic $X$-trees.  
Fix the uniform distribution over these trees.  For any fixed
choice of the random coins, an adaptive algorithm using at most $q$
queries is a ternary decision tree of depth at most $q$, and hence has
at most $3^q$ terminal transcripts.  Each terminal transcript can
correctly output at most one inducing tree.  Its average
reconstruction success under the uniform distribution is therefore
at most~$3^q/M_n$.  Averaging also over the algorithm's coins and
using success probability at least $1-\delta$ on every tree gives
$3^q/M_n\geq 1-\delta$, proving part~(a).

For part~(b), fix a four-set $C\in\binom X4$.  Construct
three phylogenetic trees that are identical outside a pendant
subtree with leaf set $C$ but realize three distinct internal
topologies within that subtree.  Any quartet query containing at most
two taxa from $C$ has the same answer in all three trees.  Call a
query \emph{touching $C$} if it contains \emph{at least three} taxa from
$C$, and let $P_C$ be the probability, over the non-adaptive
algorithm's random coins, that its query set touches~$C$.

Conditioned on a choice of coins whose query set does not
touch $C$, the three trees produce the same transcript, so the
algorithm can be correct on at most one of them.  Averaged over the
three trees, its success probability is at most 
\[
P_C+\frac{1-P_C}{3}
=\frac13+\frac{2}{3}P_C.
\]
Since the algorithm's success probability is at least~$1-\delta$ on each
tree, we derive that~$P_C\geq 1-3\delta/2$.  A fixed quartet query $Y$ touches exactly 
$1+4(n-4)=4n-15$ four-sets $C$ since either $C=Y$, or $C$ contains three taxa of~$Y$ 
and one taxon outside~$Y$.  Summing the touching probabilities
over all $C\in\binom X4$ therefore yields
\[
\left(1-\frac{3\delta}{2}\right)\binom n4
\leq
\sum_{C\in\binom X4}P_C
\leq q(4n-15),
\]
in which the second inequality follows from the union bound. This proves part~(b).  
The same pendant-subtree obstruction
underlies the non-adaptive quartet-reconstruction lower bound of
Arvanitakis et al.~\cite{ArvanitakisEtAl26}. 
%The counting argument above records the guarantee for randomized algorithms explicitly. 
\end{proof}

\begin{cor}[Tight reconstruct-or-reject complexity]
\label{cor:reconstruct-or-reject-tight}
For $\varepsilon,\delta\in(0,1/4]$ and $n$ sufficiently
large as a function of $\varepsilon$ and $\delta$, the worst-case query complexity of one-sided reconstruct-or-reject
algorithms is
\[
\Theta\!\left(
n\log n+
\frac{\ln(1/\delta)}{\ln(1/(1-\varepsilon))}
\right)
\]
in the adaptive model and 
\[
\Theta\!\left(
n^3+ \frac{\ln(1/\delta)}{\ln(1/(1-\varepsilon))}
\right) 
\]
in the non-adaptive model.  Equivalently, these bounds are
$\Theta(n\log n+\varepsilon^{-1}\log(1/\delta))$ and
$\Theta(n^3+\varepsilon^{-1}\log(1/\delta))$, respectively. 
The implicit constants in the $\Theta(\cdot)$ notation are 
independent of~$n$, $\varepsilon$, and~$\delta$.
\end{cor}

\begin{proof}
The upper bounds follow from
Theorem~\ref{thm:learn-audit} and
Corollary~\ref{cor:cached-anchor} after outputting $\widehat T$ on
acceptance.  The reconstruction terms follow from
Theorem~\ref{thm:reconstruction-lb}. For a finite-size verification lower bound, use
Corollary~\ref{cor:finite-one-sided} with $\gamma=1$ and choose
$n$ large enough that $\eta_n(1)\leq\delta$.
For $\varepsilon,\delta\leq1/4$,  we have
$(1-t)^{-1}\leq (1-2\varepsilon)^{-1}$
whenever $0\leq t\leq 2\varepsilon$. Hence,
\[
-\ln(1-2\varepsilon)
=\int_0^{2\varepsilon}\frac{1}{1-t}\,dt
\leq \int_0^{2\varepsilon}\frac{1}{1-2\varepsilon}\,dt
=\frac{2\varepsilon}{1-2\varepsilon}
\leq 4\varepsilon,
\]
where the last inequality follows from $1-2\varepsilon\geq 1/2$.
Finally we have 
\[
q\geq\frac{\ln(1/(2\delta))}{-\ln(1-2\varepsilon)}
\geq \frac{\tfrac12\ln(1/\delta)}{4\varepsilon}
\geq \frac{\ln(1/\delta)}{8\varepsilon}.
\]
Thus the verification term has a uniform constant for the stated
parameter range.   
Combining the reconstruction and verification lower bounds,
and using $\max\{a,b\}\geq(a+b)/2$ for $a,b\geq0$, gives
\begin{itemize} 
\item $q =\Omega\!\left(
n\log n+\varepsilon^{-1}\ln(1/\delta)
\right)$ in the adaptive model, and 
\item $q =\Omega\!\left(
n^3+\varepsilon^{-1}\ln(1/\delta)
\right)$ in the non-adaptive model.
\end{itemize}
For $0<\varepsilon\leq1/4$, we have
$\varepsilon\leq\ln(1/(1-\varepsilon))
\leq 4\varepsilon/3$.
Thus the lower bounds above match the corresponding upper bounds
in both stated forms, with implicit constants independent 
of~$n$, $\varepsilon$, and $\delta$.
\iffalse
Taking the maximum of the two
lower bounds is within a factor of two of their sum.  Finally,
$\ln(1/(1-\varepsilon))=\varepsilon + \varepsilon^2/2 + O(\varepsilon^3) = \Theta(\varepsilon)$ for
$\varepsilon\leq 1/4$.
\fi
\end{proof}

\paragraph{Remark.} For fixed $\varepsilon$ and $\delta$, the corollary exhibits
a polynomial gap between adaptive and non-adaptive
reconstruct-or-reject algorithms.  This is a consequence of the
stronger output requirement and should not be interpreted as an
adaptivity separation for ordinary property testing.

%%%%%%%%%%%%%%%%%%%%%%%%%%%%%%%%%%%%%%%%%%%%%%%%%%%
\section{Concluding Remarks and Future Work}
\label{sec:concluding}
%%%%%%%%%%%%%%%%%%%%%%%%%%%%%%%%%%%%%%%%%%%%%%%%%%%

The sampled tester gives explicit polynomial dependence on inverse
distance~$1/\varepsilon$ independently of the taxon count.  
Its adaptive query bound is $O(\varepsilon^{-1}\log(2/\varepsilon)\log(1/\delta))$.
The soundness proof uses a global comparison tree determined by
the initial sample and a collision bound proved after conditioning
on the sampled union.  The complete-cache version gives
$O(\varepsilon^{-4}\log(1/\delta))$ non-adaptive queries.

The main quantitative question for adaptive ordinary testing is now
whether the factor $\log(2/\varepsilon)$ can be removed.
For constant confidence and sufficiently large $n$, the present
bounds are $\Omega(\varepsilon^{-1})$ and~$O(\varepsilon^{-1}\log(2/\varepsilon))$.
Improving the non-adaptive $O(\varepsilon^{-4})$ bound or proving a
stronger non-adaptive lower bound would clarify the value of
adaptivity for ordinary testing.  The $n\log n$ versus~$n^3$ gap
proved here concerns the stronger exact-reconstruction task.

A further structural question concerns the density of unresolved
quintets.  For $n\geq5$, put
\[
\eta_5(Q)=\Pr_{Z\sim\operatorname{Unif}(\binom X5)}
[Q|_Z\text{ is unresolved}].
\]
Quantitative bounds relating $\eta_5(Q)$ to
$\operatorname{dist}(Q,\mathcal P_{\rm tree})$ would analyze the
especially simple tester that samples independent quintets.
The sampled-tree proof already yields a polynomial-query tester
without requiring such a quintet-density estimate.  
Further interesting directions include tolerant testing between 
two positive distance thresholds and testing with missing quartet 
entries.

%%%%%%%%%%%%%%%%%%%%%%%%%%%%%%%%%%%%%%%%%%%%%%%%%%%
\section*{Statements and Declarations}
%%%%%%%%%%%%%%%%%%%%%%%%%%%%%%%%%%%%%%%%%%%%%%%%%%%

\paragraph{Funding}
This work was supported by the National Science and Technology
Council of Taiwan under grant no.~NSTC 115-2221-E-019-045-.

\paragraph{Competing interests}
The author has no relevant financial or non-financial interests
to disclose.

\paragraph{Data availability}
No datasets were generated or analyzed during the current study.

\paragraph{Author contributions}
Chuang-Chieh Lin conceived the study, developed and verified the
results, and wrote and revised the manuscript.

\paragraph{Use of generative AI and AI-assisted technologies}
The author worked on the original problem for couple of months and obtained 
initial testability result and a lower bound. 
During subsequent revisions, the author used GPT-6 Sol (OpenAI) to assist with 
the preliminary formulation of the explicit tester and the drafting of the proof 
of Theorem~\ref{thm:polynomial-tester}, including its supporting lemmas and 
related corollaries. 
The author subsequently verified the mathematical arguments and proofs, 
and revise the exposition to improve readability and explain the underlying 
intuition. The tool also assisted with picture drawing and language editing.  
The author takes full responsibility for the correctness of the results 
and all other content in the paper.

%%%%%%%%%%%%%%%%%%%%%%%%%%%%%%%%%%%%%%%%%%%%%%%%%%%%%%%%%%%%%%%%%%%%%%%%

%%% The next two lines define, first, the bibliography style to be 
%%% applied, and, second, the bibliography file to be used.

\bibliographystyle{elsarticle-num}
\bibliography{tqc_note}

%%%%%%%%%%%%%%%%%%%%%%%%%%%%%%%%%%%%%%%%%%%%%%%%%%%%%%%%%%%%%%%%%%%%%%%%

\end{document}